\documentclass[10pt,oneside,english]{amsart}%
\usepackage{amsfonts}
\usepackage[T1]{fontenc}
\usepackage[latin9]{inputenc}
\usepackage{amsthm}
\usepackage{amssymb}
\usepackage[a4paper,hmargin=1in,vmargin=1in]{geometry}
\usepackage{amsmath}
\usepackage{amssymb}
\usepackage{xcolor}
\usepackage{babel}
\usepackage{graphicx}
\usepackage{enumerate}%
\providecommand{\U}[1]{\protect\rule{.1in}{.1in}}
\makeatletter
\numberwithin{equation}{section}
\numberwithin{figure}{section}
\let\pa\partial

\newtheorem{theorem}{Theorem}
\newtheorem{lemma}[theorem]{Lemma}
\newtheorem{proposition}[theorem]{Proposition}

\let\ga=\gamma
\let\de=\delta

\let\la=\lambda

\let\pa=\partial

\let\om=\omega

\begin{document}
\title[Landau equation]{Smoothing effects and decay estimate of the solution of the linearized two species Landau equation}
\author{Yu-Chu Lin}
\address{Department of Mathematics, National Cheng Kung University}
\email{yuchu@mail.ncku.edu.tw }
\author{Haitao Wang}
\address{Institute of Natural Sciences and School of Mathematical Sciences, Shanghai
Jiao Tong University}
\email{haitaowang.math@gmail.com}
\author{Kung-Chien Wu}
\address{Department of Mathematics, National Cheng Kung University and National Center
for Theoretical Sciences, National Taiwan University}
\email{kungchienwu@gmail.com}
\thanks{The first author is supported by the Ministry of Science and Technology under
the grant MOST 105-2115-M-006-002-. Part of the project was done while the
second author was in Institute of Mathematics, Academia Sinica, Taiwan. The
third author is supported by the Ministry of Science and Technology under the
grant 104-2628-M-006-003-MY4 and National Center for Theoretical Sciences.}

\begin{abstract}
We study the Landau equation for a mixture of two species in the whole space,
with initial condition of one species near a vacuum and the other near a
Maxwellian equilibrium state. For the linearized level, without any smoothness assumption on the initial
data, it is shown that the solution becomes smooth instantaneously in both the
space and momentum variables. Moreover, the large time behavior of the
solution is also obtained.

\end{abstract}
\keywords{Landau equation; Large time behavior; Smoothing effects.}
\subjclass[2010]{35Q20; 82C40.}
\maketitle




\section{Introduction}

\subsection{The Model}

This paper is concerned with the Cauchy problem for a system of Landau
equations describing collisions in an ideal plasma mixture. The mixture is
constituted by two species with mass $(m_{A}, m_{B})$ and is described by
density functions $(F_{A}(t, x, p), F_{B}(t,x,p))$ defined in the phase-space
of position and momentum. It takes the form of the following system:
\begin{equation}
\left\{
\begin{array}
[c]{l}%
\displaystyle\pa_{t}F_{A}+\frac{1}{m_{A}}p\cdot\nabla_{x}F_{A}=Q^{AA}%
(F_{A},F_{A})+Q^{AB}(F_{A},F_{B})\,,\\[4mm]%
\displaystyle\pa_{t}F_{B}+\frac{1}{m_{B}}p\cdot\nabla_{x}F_{B}=Q^{BB}%
(F_{B},F_{B})+Q^{BA}(F_{B},F_{A})\,.
\end{array}
\right.  \label{mix.1.a}%
\end{equation}
The right hand side consists of the usual collision terms which for
$X,Y\in\{A,B\}$ are given by:
\[
Q^{XY}(F_{X},F_{Y})=\nabla_{p}\cdot\Big\{\int_{{\mathbb{R}}^{3}}\Phi
^{X,Y}\left(  \frac{p}{m_{X}}-\frac{p_{\ast}}{m_{Y}}\right)  \big[F_{Y}%
(p_{\ast})\nabla_{p}F_{X}(p)-F_{X}(p)\nabla_{p_{\ast}}F_{Y}(p_{\ast
})\big]dp_{\ast}\Big\}\,,
\]
here
\[
\Phi^{X,Y}(z)=\frac{m_{X}m_{Y}}{m_{X}+m_{Y}}\left[  I_{3}-\frac{z\otimes
z}{|z|^{2}}\right]  \varphi(|z|)\,,
\]
the potential $\varphi(|z|)=|z|^{\ga+2}$. We assume throughout this paper that
$\ga\in[-2,1]$.

We can find the following invariant properties of the collision operator (see
\cite{[Gualdani]}):
\begin{equation}
\left\{
\begin{array}
[c]{l}%
\displaystyle\int_{{\mathbb{R}}^{3}}\big\{1,p,|p|^{2}\big\}Q^{XX}(F_{X}%
,F_{X})dp=0\quad\hbox{for}\quad X\in\{A,B\}\,,\\[3mm]%
\displaystyle\int_{{\mathbb{R}}^{3}}Q^{AB}(F_{A},F_{B})dp=\int_{{\mathbb{R}%
}^{3}}Q^{BA}(F_{B},F_{A})dp=0\,,\\[3mm]%
\displaystyle\int_{{\mathbb{R}}^{3}}p\left[  Q^{AB}(F_{A},F_{B})+Q^{BA}%
(F_{B},F_{A})\right]  dp=0\,,\\[3mm]%
\displaystyle\int_{{\mathbb{R}}^{3}}\frac{|p|^{2}}{2}\left[  \frac{1}{m_{A}%
}Q^{AB}(F_{A},F_{B})+\frac{1}{m_{B}}Q^{BA}(F_{B},F_{A})\right]  dp=0\,.
\end{array}
\right.  \label{invQ}%
\end{equation}
This means that the system $\left(  \ref{mix.1.a}\right)  $ is possessed of
the conservation of mass, total momentum and total energy:
\begin{equation}
\left\{
\begin{array}
[c]{l}%
\displaystyle\frac{d}{dt}\int_{{\mathbb{R}}^{3}\times{\mathbb{R}}^{3}}%
F_{X}dpdx=0\quad\hbox{for}\quad X\in\{A,B\}\,,\\[3mm]%
\displaystyle\frac{d}{dt}\int_{{\mathbb{R}}^{3}\times{\mathbb{R}}^{3}}p\left[
F_{A}+F_{B}\right]  dpdx=0\,,\\[3mm]%
\displaystyle\frac{d}{dt}\int_{{\mathbb{R}}^{3}\times{\mathbb{R}}^{3}}%
\frac{|p|^{2}}{2}\left[  \frac{1}{m_{A}}F_{A}+\frac{1}{m_{B}}F_{B}\right]
dpdx=0\,.
\end{array}
\right.  \label{conservation}%
\end{equation}

\subsection{The Linearized Problem}

Based on the idea introduced by \cite{[Sotirov]}, we consider the following
perturbation:
\[
\left\{
\begin{array}
[c]{l}%
\displaystyle F_{A}=\sqrt{M_{A}}f_{A}\,,\\[3mm]%
\displaystyle F_{B}=M_{B}+\sqrt{M_{B}}f_{B}\,,
\end{array}
\right.
\]
with the initial conditions of $F_{A}$ and $F_{B}$ satisfying
\[
\left\{
\begin{array}
[c]{l}%
\displaystyle F_{A}(0,x,p)=\sqrt{M_{A}}g_{in}(x,p)\,,\\[3mm]%
\displaystyle F_{B}(0,x,p)=M_{B}+\sqrt{M_{B}}h_{in}(x,p),
\end{array}
\right.
\]
where $M_{A}$ and $M_{B}$ are the Maxwellian states. Here, we fix the
Maxwellian state of species $B$ as the standard Gaussian, i.e.,
\[
M_{B}(p)=\frac{1}{(2\pi)^{3/2}}\exp\Big(\frac{-|p|^{2}}{2}\Big)\,,
\]
and then $M_{A}$, the Maxwellian state of species $A$, is chosen uniquely as%
\[
M_{A}(p)=\frac{1}{(2\pi m_{A}/m_{B})^{3/2}}\exp\Big(-\frac{m_{A}}{m_{B}}%
\frac{|p|^{2}}{2}\Big)\,,
\]
for which the condition
\[
Q^{AB}(M_{A},M_{B})=Q^{BA}(M_{B},M_{A})=0,
\]
is valid. Therefore, under this choice, the two Maxwellian states satisfy
\begin{equation}
\left\{
\begin{array}
[c]{l}%
\displaystyle Q^{AA}(M_{A},M_{A})=0,\quad Q^{BB}(M_{B},M_{B})=0\,,\\[3mm]%
\displaystyle Q^{AB}(M_{A},M_{B})=0,\quad Q^{BA}(M_{B},M_{A})=0\,.
\end{array}
\right.  \label{mix.1.b}%
\end{equation}

It is easy to check that the system for $(f_{A},f_{B})$ satisfies
\[
\left\{
\begin{array}
[c]{l}%
\displaystyle \pa_{t}f_{A}+\frac{1}{m_{A}}p\cdot\nabla_{x} f_{A}=L_{AB}%
f_{A}+\Gamma^{AA}(f_{A},f_{A})+\Gamma^{AB}(f_{A},f_{B})\,,\\[3mm]%
\displaystyle \pa_{t}f_{B}+\frac{1}{m_{B}}p\cdot\nabla_{x} f_{B}=L_{BB}%
f_{B}+\Gamma^{BB}(f_{B},f_{B})+L_{BA}f_{A}+\Gamma^{BA}(f_{B},f_{A})\,,
\end{array}
\right.
\]
where
\[
\left\{
\begin{array}
[c]{l}%
\displaystyle L_{AB}f_{A}=\frac{1}{\sqrt{M_{A}}}Q^{AB}(f_{A}\sqrt{M_{A}}%
,M_{B})\,,\\[3mm]%
\displaystyle L_{BA}f_{A}=\frac{1}{\sqrt{M_{B}}}Q^{BA}(M_{B}, f_{A}\sqrt
{M_{A}})\,,\\[3mm]%
\displaystyle L_{BB}f_{B}=\frac{1}{\sqrt{M_{B}}}\Big(Q^{BB}(M_{B}, f_{B}%
\sqrt{M_{B}})+Q^{BB}( f_{B}\sqrt{M_{B}}, M_{B})\Big)\,,\\[3mm]%
\displaystyle \Gamma^{XY}(f_{X},f_{Y})=\frac{1}{\sqrt{M_{X}}}Q^{XY}(f_{X}%
\sqrt{M_{X}}, f_{Y}\sqrt{M_{Y}})\,.
\end{array}
\right.
\]

If we assume the total mass of species $A$ and the perturbations in species $B$
are sufficiently small, in general, the basic time asymptotic behavior of the
nonlinear solution is governed primarily by the linearized equation (up to a large time scale), since the nonlinear term is quadratic.
In this regard, we will consider the linearized problem as follows:
\begin{equation}
\left\{
\begin{array}
[c]{l}%
\displaystyle\pa_{t}g+\frac{1}{m_{A}}p\cdot\nabla_{x}g=L_{AB}g\,,\\[3mm]%
\displaystyle\pa_{t}h+\frac{1}{m_{B}}p\cdot\nabla_{x}h=L_{BB}h+L_{BA}g\,.
\end{array}
\right.  \label{mix.2.a}%
\end{equation}

In this paper, we shall study the regularization effect and large time
behavior of the system (\ref{mix.2.a}).

\subsection{Review of Previous Works}

Let us review the mixtures of the Boltzmann equation first. The studies of gas
mixtures in terms of the Boltzmann equations were firstly introduced by
\cite{[Takata]}. Many interesting physical problems such as \textquotedblleft
ghost effects\textquotedblright\ and \textquotedblleft Knudsen
layer\textquotedblright\ for gas mixtures have been investigated in
\cite{[Aoki],[Sone]}. From a mathematical point of view, if we assume that all
species are close to equilibrium, the explicit spectrum gap estimate and the
hypocoercivity for a linearized multi-species Boltzmann system in the torus
can be found in \cite{[Mouhot]}, and exponential decay towards equilibrium in
general function spaces was done by Briant \cite{[Marc1]} and Braint and Daus
\cite{[Marc2]}. Moreover, the global existence and stability of mild solutions
to the Boltzmann system were completed by Ha, Noh and Yun in \cite{[Ha]}. This
research is also interested in phenomena related to vapor-vapor mixtures. This corresponds to the mathematical formulation that one species is
near a vacuum and the other is near a Maxwellian equilibrium state. A
qualitative-quantitative mathematical analysis in the whole space case was
introduced by Sotirov and Yu \cite{[Sotirov]}, and the torus case can be found
in \cite{QAM}.

For the mono species Landau equation, we refer to Alexandre and Villani
\cite{[2]} for existence of renormalized solutions, to Desvillette and Villani
\cite{[18]} for conditional almost exponential convergence towards equilibrium
and to a recent work by Carrapatoso, Tristani and Wu \cite{[10]} for
exponential decay towards equilibrium when initial data are close enough to
equilibrium. Moreover, Guo \cite{[Guo]} and Strain and Guo \cite{[31], [32]}
developed an existence and convergence towards equilibrium theory based on
energy methods for initial data close to the equilibrium state in some Sobolev
norms. Recently the set of initial data for which this theory is valid has been
enlarged by Carrapatoso and Mischler \cite{[9]} via a linearization method.

For the multi species Landau equation, we refer to Chapter 4 in \cite{[26]}
for a physical derivation and discussion. However, on the mathematical side,
there are very few results. One can only find the work done by Gualdani and
Zamponi \cite{[Gualdani]}, in which all species are assumed to be close to
equilibrium; the explicit spectrum gap estimate and the hypocoercivity for a
linearized multi-species Landau system in the torus were obtained.

In this paper, we are concerned with the Landau equation in the whole space,
with initial condition of one species near a vacuum and the other near a
Maxwellian equilibrium state, and study the regularization effect and large
time behavior of solutions of this Cauchy problem.

\subsection{Main theorem}

Before the presentation of the main theorem, let us define some notations in
this paper. We denote $\left\langle p\right\rangle ^{s}=(1+|p|^{2})^{s/2}$,
$s\in{\mathbb{R}}$. For the microscopic variable $p$, we denote
\[
|f|_{L_{p}^{2}}=\Big(\int_{{\mathbb{R}}^{3}}|f|^{2}dp\Big)^{1/2}\,,
\]
and the weighted norms $|f|_{L_{p}^{2}(m)}$ and $|f|_{L_{\theta}^{2}}$
respectively by
\[
|f|_{L_{p}^{2}(m)}=\Big(\int_{{\mathbb{R}}^{3}}|f|^{2}mdp\Big)^{1/2}%
\,,\quad|f|_{L_{\theta}^{2}}=\Big(\int_{{\mathbb{R}}^{3}}|f|^{2}\left\langle
p\right\rangle ^{2\theta}dp\Big)^{1/2}\,.
\]
The $L_{p}^{2}$ inner product in ${\mathbb{R}}^{3}$ will be denoted by
$\big<\cdot,\cdot\big>_{p}$. For the space variable $x$, we have similar
notations. In fact, $L_{x}^{2}$ is the classical Hilbert space with norm
\[
|f|_{L_{x}^{2}}=\Big(\int_{{\mathbb{R}}^{3}}|f|^{2}dx\Big)^{1/2}\,.
\]
We denote the sup norm as
\[
|f|_{L_{x}^{\infty}}=\sup_{x\in{\mathbb{R}}^{3}}|f(x)|\,.
\]
The standard vector product will be denoted by $(\cdot,\cdot)$. For any vector
function $u\in L_{p}^{2}$, $\mathbb{P}(p)u$ denotes the orthogonal projection
to the direction of vector $p$, i.e.,
\[
\mathbb{P}(p)u=\frac{u\cdot p}{|p|^{2}}p\,.
\]
For the Landau equation, the natural norm in $p$ is $|\cdot|_{L_{\sigma}^{2}}%
$, which is defined by
\[
|f|_{L_{\sigma}^{2}}^{2}=|\left\langle p\right\rangle ^{\frac{\ga+2}{2}%
}f|_{L_{p}^{2}}^{2}+|\left\langle p\right\rangle ^{\frac{\ga}{2}}%
\mathbb{P}(p)\nabla_{p}f|_{L_{p}^{2}}^{2}+\big|\left\langle p\right\rangle
^{\frac{\ga+2}{2}}\big[I_{3}-\mathbb{P}(p)\big]\nabla_{p}f\big|_{L_{p}^{2}%
}^{2}\,,
\]
and the weighted norm by
\[
|f|_{L_{\sigma,\theta}^{2}}^{2}=|\left\langle p\right\rangle ^{\frac{\ga+2}%
{2}+\theta}f|_{L_{p}^{2}}^{2}+|\left\langle p\right\rangle ^{\frac{\ga}%
{2}+\theta}\mathbb{P}(p)\nabla_{p}f|_{L_{p}^{2}}^{2}+\big|\left\langle
p\right\rangle ^{\frac{\ga+2}{2}+\theta}\big[I_{3}-\mathbb{P}(p)\big]\nabla
_{p}f\big|_{L_{p}^{2}}^{2}\,.
\]
Moreover, we define
\[
\Vert f\Vert_{L^{2}}^{2}=\int_{{\mathbb{R}}^{3}}|f|_{L_{p}^{2}}^{2}%
dx\,,\quad\Vert f\Vert_{L^{2}(m)}^{2}=\int_{{\mathbb{R}}^{3}}|f|_{L_{p}%
^{2}(m)}^{2}dx\,,\quad\Vert f\Vert_{L_{\theta}^{2}}^{2}=\int_{{\mathbb{R}}%
^{3}}|f|_{L_{\theta}^{2}}^{2}dx\,,
\]%
\[
\Vert f\Vert_{L_{\sigma}^{2}}^{2}=\int_{{\mathbb{R}}^{3}}|f|_{L_{\sigma}^{2}%
}^{2}dx\,,\quad\Vert f\Vert_{L_{\sigma,\theta}^{2}}^{2}=\int_{{\mathbb{R}}%
^{3}}|f|_{L_{\sigma,\theta}^{2}}^{2}dx\,,
\]
and
\[
\Vert f\Vert_{L_{x}^{\infty}L_{p}^{2}}=\sup_{x\in{\mathbb{R}}^{3}}%
|f|_{L_{p}^{2}}\,,\quad\Vert f\Vert_{L_{x}^{1}L_{p}^{2}}=\int_{{\mathbb{R}%
}^{3}}|f|_{L_{p}^{2}}dx\,.
\]
Finally, we define the high order Sobolev norm: Let $k\in{\mathbb{N}}$ and
$k_{1},k_{2}\geq0$,
\[
\left\Vert f\right\Vert _{H^{k}(m)}=\sum_{k_{1}+k_{2}=k}\left\Vert
\partial_{x}^{k_{1}}\pa_{p}^{k_{2}}f\right\Vert _{L^{2}(m)}\,.
\]
The domain decomposition plays an important role for the coercivity of the
collision operators, hence we define a cut-off function $\chi:{\mathbb{R}%
}\rightarrow{\mathbb{R}}$, which is a smooth non-increasing function,
$\chi(s)=1$ for $s\leq1$, $\chi(s)=0$ for $s\geq2$ and $0\leq\chi\leq1$.
Moreover, we define $\chi_{R}(s)=\chi(s/R)$.

For simplicity of notations, hereafter, we abbreviate \textquotedblleft{ $\leq
C$ }\textquotedblright\ to \textquotedblleft{ $\lesssim$ }\textquotedblright,
where $C$ is a positive constant depending only on fixed numbers.

In the following, we state our main result.

\begin{theorem}
\label{MAIN}Let $k\ $and$\ \ell\ $be nonnegative integers, and assume that the
initial conditions $(g_{in},h_{in})\in L_{x}^{1}L_{p}^{2}\cap L^{2}%
(w_{k+\ell+2})$. Then the solutions of $(\ref{mix.2.a})$ have the following
estimates:
\[
\Vert\nabla_{p}^{\ell}\nabla_{x}^{k}g\Vert_{L^{2}}\lesssim(1+t)^{-(3+k)/2}%
\left(  \left\Vert g_{in}\right\Vert _{L_{x}^{1}L_{p}^{2}}+\left\Vert
g_{in}\right\Vert _{L^{2}\left(  w_{k+\ell+2}\right)  }\right)  \,,
\]
and
\[
\Vert\nabla_{p}^{\ell}\nabla_{x}^{k}h\Vert_{L^{2}}\lesssim(1+t)^{-(3+k)/2}%
\left(  \left\Vert (g_{in},h_{in})\right\Vert _{L_{x}^{1}L_{p}^{2}}+\left\Vert
(g_{in},h_{in})\right\Vert _{L^{2}\left(  w_{k+\ell+2}\right)  }\right)  \,.
\]
Here
\[
w_{n}\equiv%
\begin{cases}
1, & \ga\in\left[  0,1\right]  ,\\
\left\langle p\right\rangle ^{|\ga|n}, & \ga\in\left[  -2,0\right)  .
\end{cases}
\]

\end{theorem}

Let us remark that this theorem is a simplified version, and it is covered in
more details in Theorems \ref{theorem-g} and \ref{theorem-h}.

\subsection{Method of proof and plan of the paper}

The basic principles used to develop the theory are the \textquotedblleft
separation of scales\textquotedblright\ and \textquotedblleft regularization
estimate\textquotedblright. The separation of scales is a concept created in
\cite{[LiuYu]} for a 1D Boltzmann equation with hard sphere in the whole space.
The regularization estimate for the Landau equation was constructed in \cite{[10]} initially.

In this paper, a long wave-short wave decomposition is introduced. Under this
decomposition, one can follow the spectrum analysis of the Landau collision
operator in \cite{YangYu} to decompose the solution to the linearized Landau
equation as the long wave-fluid part, the lone wave-nonfluid part and the
short wave part, which separate the scales, and hence get different decay rates
for each part.

On the other hand, the regularization estimate plays an important role in this
paper, which leads us to remove the regularity assumption on the initial
condition. Based on the basic regularization estimate established in
\cite{[10]}, we are able to construct new functionals for high order
derivatives, and then by using bootstrap process to get the smoothing
estimates in both the space variable $x$ and momentum variable $p$.

The separation of scales helps us get
different decay rate in each part, but this idea can not improve regularity.
Contrarily, the regularization estimate helps us improve regularity, but this can not get
time decay estimates. Fortunately, if putting these two ideas together, by
interpolation argument, one can get a smoothing effect and time decay
simultaneously. To the best of our knowledge, this thinking is itself new for
Landau kinetic equations.

In the following we point out some significant points of our method:

\begin{itemize}
\item The problem we consider here is the whole space case, and hence we
expect the algebraic decay estimate of the solutions rather than the
exponential decay in the torus case \cite{[Marc1], [Marc2], [10], [Mouhot],
[Gualdani], [31], [32]}. Under our linearization (one species near a vacuum
and the other near equilibrium), the system is decoupled, so that one can
estimate each equation $(g,h)$ separately. For solution $g$, as we mentioned
above, one can get smoothing effect and time decay
simultaneously by interpolation argument.

For solution $h$, we solve it by using the Duhamel principle and
treating $g$ as a source term. This method usually results in that the decay
rate of $h$ is slower than that of $g$. Indeed, the characteristics for the
macroscopic equations for $h$ coincide with the direction of mass diffusion
for $g$. Thus, there is a possible wave resonance, which leads to the slower
decay of solution $h$. On the microscopic level, the resonance is from the
small divisor caused by the closeness between different wave frequencies
(i.e., $\lambda(\eta)$ in Lemma \ref{SpecAB} and $\sigma_{j}%
(\eta),~j=2,3,4$ in Lemma \ref{SpecBB}). It was shown in \cite{[Sotirov]} that
for a 1D Boltzmann gas mixture there is only one characteristic for $h$
coinciding with $g$, which is the diffusion sound wave carrying mass and
energy. The resonance problem was resolved by considering a microscopic
cancellation representing the conservation of mass and total energy. For our
3D case, there are other two characteristics for $h$ having the same direction
as mass diffusion for $g$ in addition to the diffusion sound wave. Their
resonances are resolved by other microscopic cancellations from the
conservation of total momentum and orthogonality between propagating
directions (see the proof of Proposition \ref{long-wave-h} for details). These
crucial cancellation properties lead $h$ to have the same decay rate as $g$.
It would also be interesting to understand the detailed wave propagation of
$(g,h)$, and they are postponed in our future studies.

\item The method to prove high order regularization estimates in this paper is
interesting and itself new. In the paper \cite{[10]}, they get a first order
regularization estimate for the semigroup $e^{t\mathcal{L}_{AB}}$ in both the
$x$ and $p$ variables (see Lemma \ref{regularization}), where $\mathcal{L}%
_{AB}$ contains the transport part $p\cdot\nabla_{x}$ and the ellipticity part
$\Lambda_{AB}$, and the regularization mechanism is their combined effect. In
order to improve the high order $x$ and $p$ regularities based on this result,
we employ different arguments respectively. For the space variable $x$, we
design a Picard type iteration, which treats $\mathcal{L}_{AB}$ as an anstaz
and the operator $K_{AB}$ as a source term. The first several terms in the
iteration contain the most singular part of the solution.
Thereupon, after finite
iterations (depending on the differentiation order), we can extract
the singular waves and then the remainder part will become smooth in
$x$ variable for all time (see Lemma \ref{x-regularityST}).
For the
momentum $p$, we construct a new functional including high order $p$
derivatives and all lower order $x$, $p$ and mixed derivatives, where the
lower order derivatives are used to take care of the nonzero commutator
$[p\cdot\nabla_{x},\nabla_{p}]$. By making use of $x$ regularity obtained
previously, and choosing suitable combinations of derivatives, one can get
high order regularization estimates in $p$ variable for the semigroup
$e^{t\mathcal{L}_{AB}}$. Then the smoothing effect of the solution in the $p$
variable will be obtained inductively (see subsection \ref{sub-p-regularity}).
\end{itemize}

The paper is organized as follows: we list some properties of the linearized
collision operator in section 2; then we prove the estimates of $g$ and $h$ in
sections 3 and 4 respectively.


\section{Preliminaries}

The basic structures of the linearized Landau collision operator are well
known \cite{[Degond]}. Here we take a closer look at these structures with
more modifications (see Lemma \ref{null}-Lemma \ref{LBA}). After that, Lemma
\ref{energy-Guo} is done by Guo \cite{[Guo]}, which would be of great help to
us regarding the regularization estimate. On the other hand, the separation of
scales relies on the spectrum analysis of the operator $L_{XY}^{\eta}%
=-\frac{i\eta\cdot p}{m_{X}}+L_{XY}$ in the classical Hilbert space $L_{p}%
^{2}$, where $(X,Y)=(A,B)$ or $(B,B)$; this is why our argument can not
include the case of very soft potentials and Coulomb potential. Following the systematic procedure
established by Yang and Yu \cite{YangYu}, we can easily derive the spectrum of
the operator $L_{XY}^{\eta}.$ Moreover, we demonstrate that when $\left\vert
\eta\right\vert $ is small, the corresponding eigenfunctions are not only
smooth but also decays faster than any polynomial in $p.$

\begin{lemma}
\label{null} \label{Null}\textrm{(The null spaces of $L_{BB}$ and $L_{AB}$ )}
The null space of $L_{BB}$ is a five-dimensional vector space with an
orthonormal basis $\{\chi_{i}\}_{i=0}^{4}$, where
\[
\{\chi_{0},\chi_{j},\chi_{4}\}=\big\{\sqrt{M_{B}},p_{j}\sqrt{M_{B}},\frac
{1}{\sqrt{6}}(|p|^{2}-3)\sqrt{M_{B}}\big\}\,,\quad j=1,2,3\,.
\]
On the other hand, the null space of $L_{AB}$ is a one-dimensional vector
space $\{E_{D}\}$, where $E_{D}=M_{A}^{1/2}$.
\end{lemma}

\begin{lemma}
\label{decom} \label{Grad} \textrm{(Decomposition)} Let $(X,Y)=(B,B)$ or
$(A,B)$. Then \newline\noindent\textrm{(i)}\quad the collision operator
$L_{XY}$ consists of an elliptic type operator $\widetilde{\Lambda}^{XY}$ and
an integral operator $\widetilde{K}^{XY}$:
\[
L_{XY}f=-\widetilde{\Lambda}^{XY}f+\widetilde{K}^{XY}f\,,
\]
where
\[
\widetilde{\Lambda}^{XY}f=-\nabla_{p}\cdot\left[  \sigma^{XY}\nabla
_{p}f\right]  +\left(  \frac{1}{4}\frac{m_{Y}^{2}}{m_{X}^{2}}\left(
\sigma^{XY}p,p\right)  -\frac{1}{2}\frac{m_{Y}}{m_{X}}\nabla_{p}\cdot\left[
\sigma^{XY}p\right]  \right)  f,
\]%
\[
\widetilde{K}^{XY}f=\left\{
\begin{array}
[c]{ll}%
\int_{{\mathbb{R}}^{3}}M_{B}^{-1/2}\left(  p\right)  M_{B}^{-1/2}\left(
p_{\ast}\right)  \nabla_{p}\cdot\nabla_{p_{\ast}}\cdot Z^{XY}\left(
p,p_{\ast}\right)  f\left(  p_{\ast}\right)  dp_{\ast},%
\vspace{3mm}%
& (X,Y)=(B,B)\\
0, & (X,Y)=(A,B)
\end{array}
\right.
\]
with
\[
\sigma^{XY}\left(  p\right)  =\int_{{\mathbb{R}}^{3}}\Phi^{X,Y}\left(
\frac{p}{m_{X}}-\frac{p_{\ast}}{m_{Y}}\right)  M_{Y}\left(  p_{\ast}\right)
dp_{\ast},
\]%
\[
Z^{XY}\left(  p,p_{\ast}\right)  =M_{X}\left(  p\right)  M_{Y}\left(  p_{\ast
}\right)  \Phi^{X,Y}\left(  \frac{p}{m_{X}}-\frac{p_{\ast}}{m_{Y}}\right)  .
\]
\newline\noindent\textrm{(ii)}\quad The spectrum of $\sigma^{XY}\left(
p\right)  $ consists of a simple eigenvalue $\lambda_{1}\left(  p\right)  $
associated with the eigenvector $p$, and a double eigenvalue $\lambda
_{2}\left(  p\right)  $ associated with the eigenspace $p^{\bot}.$ Furthermore%
\[
\lambda_{1}^{XY}\left(  p\right)  =\frac{m_{X}m_{Y}}{m_{X}+m_{Y}}\frac
{1}{m_{Y}^{\gamma+2}}\int_{\mathbb{R}^{3}}\left(  1-\left(  \frac
{p}{\left\vert p\right\vert },\frac{p_{\ast}}{\left\vert p_{\ast}\right\vert
}\right)  ^{2}\right)  M_{Y}\left(  \frac{m_{Y}}{m_{X}}p-p_{\ast}\right)
\left\vert p_{\ast}\right\vert ^{\gamma+2}dp_{\ast},
\]%
\[
\lambda_{2}^{XY}\left(  p\right)  =\frac{m_{X}m_{Y}}{m_{X}+m_{Y}}\frac
{1}{m_{Y}^{\gamma+2}}\int_{\mathbb{R}^{3}}\left(  1+\left(  \frac
{p}{\left\vert p\right\vert },\frac{p_{\ast}}{\left\vert p_{\ast}\right\vert
}\right)  ^{2}\right)  M_{Y}\left(  \frac{m_{Y}}{m_{X}}p-p_{\ast}\right)
\left\vert p_{\ast}\right\vert ^{\gamma+2}dp_{\ast},
\]
with%
\[
\underset{\left\vert p\right\vert \rightarrow0}{\lim\inf}\lambda_{1}%
^{XY}\left(  p\right)  >0,\ \ \ \underset{\left\vert p\right\vert
\rightarrow0}{\lim\inf}\lambda_{2}^{XY}\left(  p\right)  >0,
\]
and as$\ \left\vert p\right\vert \rightarrow\infty,$
\[
\lambda_{1}^{XY}\left(  p\right)  \sim\frac{m_{X}m_{Y}}{m_{X}+m_{Y}}\frac
{2}{m_{Y}^{\gamma+2}}\left(  \frac{m_{Y}}{m_{X}}\left\vert p\right\vert
\right)  ^{\gamma},\ \ \ \ \ \lambda_{2}^{XY}\left(  p\right)  \sim\frac
{m_{X}m_{Y}}{m_{X}+m_{Y}}\frac{1}{m_{Y}^{\gamma+2}}\left(  \frac{m_{Y}}{m_{X}%
}\left\vert p\right\vert \right)  ^{\gamma+2}.
\]
Immediately, we have
\begin{align*}
\nabla_{p}\cdot\sigma^{XY}(p)  &  =-\frac{m_{Y}^{2}}{m_{X}^{2}}\sigma
^{XY}(p)p=-\frac{m_{Y}^{2}}{m_{X}^{2}}\la_{1}^{XY}(p)p\,,\quad(p,\sigma
^{XY}(p)p)=\la_{1}^{XY}(p)|p|^{2}\,,\\
(u,\sigma^{XY}(p)u)  &  =\la_{1}^{XY}(p)|\mathbb{P}(p)u|^{2}+\la_{2}%
^{XY}(p)\big|\big[I_{3}-\mathbb{P}(p)\big]u\big|^{2}\\
&  \geq c_{0}\left\{  \left\langle p\right\rangle ^{\ga}\left\vert
\mathbb{P}(p)u\right\vert ^{2}+\left\langle p\right\rangle ^{\ga+2}\left\vert
\left[  I_{3}-\mathbb{P}(p)\right]  u\right\vert ^{2}\right\}  \,.
\end{align*}
\textrm{(iii)} For any $k\in{\mathbb{N}}$, we have
\[
|\nabla_{p}^{k}\sigma^{XY}(p)|\lesssim\left\langle p\right\rangle
^{\ga+2-k}\,,\quad|\nabla_{p}^{k}(\sigma^{XY}(p)p)|\lesssim\left\langle
p\right\rangle ^{\ga+1-k}\,,
\]
and as $|p|\rightarrow\infty$,
\[
|\nabla_{p}^{k}\la_{1}^{XY}(p)|\lesssim\left\langle p\right\rangle
^{\ga-k}\,,\quad|\nabla_{p}^{k}\la_{2}^{XY}(p)|\lesssim\left\langle
p\right\rangle ^{\ga+2-k}\,.
\]
\textrm{(iv)}\quad(Coercivity) Rewrite
\begin{align*}
\widetilde{\Lambda}^{XY}f  &  =-\nabla_{p}\cdot\left[  \sigma^{XY}\nabla
_{p}f\right]  +\left(  \frac{1}{4}\frac{m_{Y}^{2}}{m_{X}^{2}}\left(
\sigma^{XY}p,p\right)  -\frac{1}{2}\frac{m_{Y}}{m_{X}}\left(  \left(
\nabla_{p}\cdot\sigma^{XY},p\right)  +Tr\sigma^{XY}\right)  \right)  f\\
&  =-\nabla_{p}\cdot\left[  \sigma^{XY}\nabla_{p}f\right] \\
&  +\left(  \frac{1}{4}\lambda_{1}^{XY}\left(  p\right)  \left(  \frac{m_{Y}%
}{m_{X}}\left\vert p\right\vert \right)  ^{2}+\frac{1}{2}\lambda_{1}%
^{XY}\left(  p\right)  \frac{m_{Y}}{m_{X}}\left(  \frac{m_{Y}}{m_{X}%
}\left\vert p\right\vert \right)  ^{2}-\frac{m_{Y}}{m_{X}}\lambda_{2}%
^{XY}\left(  p\right)  -\frac{1}{2}\frac{m_{Y}}{m_{X}}\lambda_{1}^{XY}\left(
p\right)  \right)  f.
\end{align*}
Let
\[
-\Lambda^{XY}f=-\widetilde{\Lambda}^{XY}f-\varpi\chi_{R}f,\text{ }%
K^{XY}f=\widetilde{K}^{XY}f+\varpi\chi_{R}f,
\]
where $\varpi,R$ are positive constants large enough. Then
\[
\left\langle \Lambda^{XY}f,f\right\rangle _{p}\geq c_{0}\left\vert
f\right\vert _{L_{\sigma}^{2}}^{2},
\]
for some $c_{0}>0,\ $and
\[
\left\langle K^{XY}f,f\right\rangle _{p}\leq\left\vert f\right\vert
_{L_{p}^{2}}^{2}.
\]

\end{lemma}

The behavior of the operator $L_{BA}$ is different from $L_{BB}$ and $L_{AB}$,
but similar to $K^{BB}.$

\begin{lemma}
\label{LBA}The operator $L_{BA}$ is a bounded integral operator:
\[
L_{BA}f=\int_{{\mathbb{R}}^{3}}M_{B}^{-1/2}\left(  p\right)  M_{A}%
^{-1/2}\left(  p_{\ast}\right)  \nabla_{p}\cdot\nabla_{p_{\ast}}\cdot
Z^{BA}\left(  p,p_{\ast}\right)  f\left(  p_{\ast}\right)  dp_{\ast},
\]
where%
\[
Z^{BA}\left(  p,p_{\ast}\right)  =M_{B}\left(  p\right)  M_{A}\left(  p_{\ast
}\right)  \Phi^{B,A}\left(  \frac{p}{m_{B}}-\frac{p_{\ast}}{m_{A}}\right)  .
\]

\end{lemma}

\begin{lemma}
[{Lemma 6, \cite{[Guo]}}]\label{energy-Guo} Let $(X,Y)=(B,B)$ or $(A,B)$,
assume $\left\vert \alpha\right\vert \geq0$, for small $\de>0$, we have
\[
\langle\left\langle p\right\rangle ^{2\theta}\pa_{p}^{\alpha}\widetilde
{\Lambda}^{XY}f,\pa_{p}^{\alpha}\overline{f}\rangle_{p}\geq|\pa_{p}^{\alpha
}f|_{L_{\sigma,\theta}^{2}}^{2}-\de\sum_{|\overline{\alpha}|\leq|\alpha
|}|\pa_{p}^{\overline{\alpha}}f|_{L_{\sigma,\theta}^{2}}^{2}-C_{\de}|\chi
_{R}f|_{L_{p}^{2}}^{2}%
\]
and
\[
\left|  \langle\left\langle p\right\rangle ^{2\theta}\pa_{p}^{\alpha
}\widetilde{K}^{XY}f_{1},\pa_{p}^{\alpha}\overline{f}_{2}\rangle_{p}\right|
\leq|\pa_{p}^{\alpha}f_{2}|_{L_{\sigma,\theta}^{2}}\left\{  \de\sum
_{|\overline{\alpha}|\leq|\alpha|}|\pa_{p}^{\overline{\alpha}}f_{1}%
|_{L_{\sigma,\theta}^{2}}+C_{\de}|\chi_{R}f_{1}|_{L_{p}^{2}}\right\}  \,.
\]
In particular,
\begin{equation}
\langle\left\langle p\right\rangle ^{2\theta}\pa_{p}^{\alpha}L_{XY}%
f,\pa_{p}^{\alpha}\overline{f}\rangle_{p}\leq-|\pa_{p}^{\alpha}f|_{L_{\sigma
,\theta}^{2}}^{2}+\de\sum_{|\overline{\alpha}|\leq|\alpha|}|\pa_{p}%
^{\overline{\alpha}}f|_{L_{\sigma,\theta}^{2}}^{2}+C_{\de}|\chi_{R}%
f|_{L_{p}^{2}}^{2}\,. \label{Guo-1}%
\end{equation}

\end{lemma}

Taking the Fourier transform in the space variable $x,$ we consider the
operator $L_{XY}^{\eta}=-\frac{i\eta\cdot p}{m_{X}}+L_{XY}$ , where
$(X,Y)=(A,B)$ or $(B,B)$. Then the corresponding spectrum $\mathrm{Spec}%
_{XY}(\eta)$ in the classical Hilbert space $L_{p}^{2}$ are derived as follows.

\begin{lemma}
\textrm{(Spectrum of $L_{AB}^{\eta}$ \cite{YangYu})}\label{SpecAB} Given
$\de>0$ \newline\noindent\textrm{(i)}\quad there exists $\tau_{2}=\tau
_{2}(\de)>0$ such that if $|\eta|>\de$,
\begin{equation}
\hbox{\rm Spec}(\eta)\subset\{z\in\mathbb{C}:\mathrm{Re}(z)<-\tau_{2}\}\,.
\label{pre.ab.e}%
\end{equation}
\newline\noindent\textrm{(ii)}\quad If $\eta=|\eta|\om$ and $|\eta|<\de$, the
spectrum within the region $\{z\in\mathbb{C}:Re(z)>-\tau_{2}\}$ is consisting
of exactly one eigenvalue $\{\la(\eta)\}$,
\begin{equation}
\hbox{\rm Spec}(\eta)\cap\{z\in\mathbb{C}:\mathrm{Re}(z)>-\tau_{2}%
\}=\{\la(\eta)\}\,, \label{pre.ab.f}%
\end{equation}
associated with its eigenfunction $\{e_{D}(\eta)\}$. They have the expansions
\begin{equation}%
\begin{array}
[c]{l}%
\label{pre.ab.h}\displaystyle\la(\eta)=-a_{2}|\eta|^{2}+O(|\eta|^{3})\,,\\
\\
\displaystyle e_{D}(\eta)=E_{D}+iE_{D,1}|\eta|+O(|\eta|^{2})\,,
\end{array}
\end{equation}
with $a_{2}=-\left<  L_{AB}^{-1}(p\cdot\om /m_{A})E_{D},(p\cdot\om /m_{A}%
)E_{D}\right>  _{p}>0$ and $E_{D,1}=L_{AB}^{-1}(p\cdot\om /m_{A})E_{D}.$ Here
$\{e_{D}(\eta)\}$ can be normalized by $\big<e_{D}(-\eta),e_{D}(\eta
)\big>_{p}=1$.

More precisely, the semigroup $e^{(-i\eta\cdot p +L_{AB})t}$ can be decomposed
as
\begin{equation}
\displaystyle e^{(-i\eta\cdot p +L_{AB})t}f=e^{(-i\eta\cdot p +L_{AB})t}%
\Pi_{\eta}^{D\perp}f+\chi_{\{|\eta|<\de\}}e^{\la(\eta)t}\big<e_{D}%
(-\eta),f\big>_{p}e_{D}(\eta)\,, \label{pre.ab.g}%
\end{equation}
and there exist $a(\tau_{2})>0$, $\overline{a}_{2}>0$ such that $|e^{(-i\eta
\cdot p +L_{AB})t}\Pi_{\eta}^{D\perp}|_{L_{p}^{2}}\lesssim e^{-a(\tau_{2})t}$
and $|e^{\la(\eta)t}|\leq e^{-\overline{a}_{2}|\eta|^{2}t}$.

Moreover, the eigenfunction $e_{D}(p,\eta)$ is smooth and decays faster than
any polynomial in $p$, i.e., for any $\theta\geq0$, $|\alpha|\geq0$,
$|\left\langle p\right\rangle ^{\theta}\pa_{p}^{\alpha}e_{D}(p,\eta
)|_{L_{p}^{2}}\leq C_{\theta,\alpha}|e_{D}(p,\eta)|_{L_{p}^{2}}$.
\end{lemma}

\begin{proof}
The spectrum analysis can be found in \cite{YangYu}. Here it is suffices to
show that $e_{D}(p,\eta)$ as a function of $p$ is in the Schwartz space. Note
that
\[
\left(  L_{AB}-\frac{i\eta\cdot p}{m_{A}}\right)  e_{D}(p,\eta)=\lambda
(\eta)e_{D}(p,\eta)\,,
\]
for any $\theta\geq0$ and $|\alpha|>0$. Taking the inner product, we have
\[
\left\langle \left\langle p\right\rangle ^{2\theta}\left(  L_{AB}-\frac
{i\eta\cdot p}{m_{A}}\right)  e_{D}(p,\eta),e_{D}(p,-\eta)\right\rangle
_{p}=\lambda(\eta)|\left\langle p\right\rangle ^{\theta}e_{D}(p,\eta
)|_{L_{p}^{2}}^{2}\,,
\]
and
\[
\left\langle \left\langle p\right\rangle ^{2\theta}\left(  L_{AB}+\frac
{i\eta\cdot p}{m_{A}}\right)  e_{D}(p,-\eta),e_{D}(p,\eta)\right\rangle
_{p}=\lambda(-\eta)|\left\langle p\right\rangle ^{\theta}e_{D}(p,\eta
)|_{L_{p}^{2}}^{2}\,.
\]
In view of (\ref{Guo-1}),
\[
|\partial_{p}^{\alpha}e_{D}(p,\eta)|_{L_{\sigma,\theta}^{2}}^{2}\lesssim
\frac{1}{2}(|\lambda(\eta)|+|\lambda(-\eta)|)|\left\langle p\right\rangle
^{\theta}\partial_{p}^{\alpha}e_{D}(p,\eta)|_{L_{p}^{2}}^{2}+\delta
\sum_{|\overline{\alpha}|\leq|\alpha|}|\partial_{p}^{\overline{\alpha}}%
e_{D}(p,\eta)|_{L_{\sigma,\theta}^{2}}^{2}+C_{\delta}|\chi_{R}e_{D}%
(p,\eta)|_{L_{p}^{2}}^{2}\,.
\]

\end{proof}

\begin{lemma}
\textrm{(Spectrum of $L_{BB}^{\eta}$ \cite{YangYu} )}\label{SpecBB} Given
$\de>0$, \newline\noindent\textrm{(i)}\quad there exists $\tau_{1}=\tau
_{1}(\de)>0$ such that if $|\eta|>\de$,
\begin{equation}
\hbox{\rm Spec}_{BB}(\eta)\subset\{z\in\mathbb{C}:\mathrm{Re}(z)<-\tau
_{1}\}\,. \label{pre.bb.e}%
\end{equation}
\newline\noindent\textrm{(ii)}\quad If $\eta=|\eta|\om$ and $|\eta|<\de$, the
spectrum within the region $\{z\in\mathbb{C}:Re(z)>-\tau_{1}\}$ consists of
exactly five eigenvalues $\{\sigma_{j}(\eta)\}_{j=0}^{4}$,
\begin{equation}
\hbox{\rm Spec}_{BB}(\eta)\cap\{z\in\mathbb{C}:\mathrm{Re}(z)>-\tau
_{1}\}=\{\sigma_{j}(\eta)\}_{j=0}^{4}\,, \label{pre.bb.f}%
\end{equation}
associated with the corresponding eigenfunctions $\{e_{j}(\eta)\}_{j=0}^{4}$.
They have the expansions
\begin{equation}%
\begin{array}
[c]{l}%
\label{pre.bb.h}\displaystyle\sigma_{j}(\eta)=ia_{j,1}|\eta|-a_{j,2}|\eta|
^{2}+O(|\eta|^{3})\,,\\
\\
\displaystyle e_{j}(\eta)=E_{j}+O(|\eta|)\,,
\end{array}
\end{equation}
with $a_{j,2}=\left<  L_{BB}^{-1}\mathbb{P}_{1}(p\cdot\om/m_{B})E_{j}%
,\mathbb{P}_{1}(p\cdot\om/m_{B})E_{j}\right>  _{p}>0$ and
\begin{equation}
\left\{
\begin{array}
[c]{l}%
a_{01}=\frac{1}{m_{B}}\sqrt{\frac{5}{3}}\,,\quad a_{11}= -\frac{1}{m_{B}}%
\sqrt{\frac{5}{3}}\,,\quad a_{21}=a_{31}=a_{41}=0\,,\\[2mm]%
E_{0}=\sqrt{\frac{3}{10}}\chi_{0}+\sqrt{\frac{1}{2}}\om\cdot\Psi+\sqrt
{\frac{1}{5}}\chi_{4}\,,\\[2mm]%
E_{1}=\sqrt{\frac{3}{10}}\chi_{0}-\sqrt{\frac{1}{2}}\om\cdot\Psi+\sqrt
{\frac{1}{5}}\chi_{4}\,,\\[2mm]%
E_{2}=-\sqrt{\frac{2}{5}}\chi_{0}+\sqrt{\frac{3}{5}}\chi_{4}\,,\\[2mm]%
E_{3}=\om_{1}^{\perp}\cdot\Psi\,,\\[2mm]%
E_{4}=\om_{2}^{\perp}\cdot\Psi\,,
\end{array}
\right.
\end{equation}
where $\Psi=(\chi_{1},\chi_{2},\chi_{3})$, $\{ \om, \om_{1}^{\perp},
\om_{2}^{\perp}\}$ is an orthonormal basis of ${\mathbb{R}}^{3}$,
$\{e_{j}(\eta)\}_{j=0}^{4}$ can be normalized by $\big<e_{j}(-\eta),e_{l}%
(\eta)\big>_{p}=\de_{jl}$, $0\leq j,l\leq4$.

More precisely, the semigroup $e^{(-i\eta\cdot p +L_{BB})t}$ can be decomposed
as
\begin{align}
\displaystyle e^{(-i\eta\cdot p +L_{BB})t}f  &  =e^{(-i\eta\cdot p +L_{BB}%
)t}\Pi_{\eta}^{\perp}f\nonumber\label{pre.bb.g}\\
&  \quad+\chi_{\{|\eta|<\de\}}\sum_{j=1}^{3}e^{\sigma_{j}(\eta)t}%
\big<e_{j}(-\eta),f\big>_{p}e_{j}(\eta)\,,
\end{align}
and there exist $a(\tau_{1})>0$, $\overline{a}_{1}>0$ such that $|e^{(-i\eta
\cdot p +L_{BB})t}\Pi_{\eta}^{\perp}|_{L_{p}^{2}}\lesssim e^{-a(\tau_{1})t}$
and $|e^{\sigma_{j}(\eta)t}|\leq e^{-\overline{a}_{1}|\eta|^{2}t}$ for all
$0\leq j\leq4$.

Moreover, the eigenfunction $e_{j}(p,\eta)$, $0\leq j\leq4$ are smooth and
decay faster than any polynomial in $p$, i.e., for any $\theta\geq0$,
$|\alpha|\geq0$, $|\left\langle p\right\rangle ^{\theta}\pa_{p}^{\alpha}%
e_{j}(p,\eta)|_{L_{p}^{2}}\leq C_{\theta,\alpha}|e_{j}(p,\eta)|_{L_{p}^{2}}$.
\end{lemma}

\section{The result for $g$}

In accordance with the theory of the semi-groups, the solution of the
linearized Landau equation
\[
\partial_{t}g+\frac{1}{m_{A}}p\cdot\nabla_{x}g=L_{AB}g,\ \ g\left(
0,x,p\right)  =g_{in}\left(  x,p\right)  ,
\]
can be represented by
\[
g\left(  t,x,p\right)  =\mathbb{G_{AB}}^{t}g_{in}=\int_{\mathbb{R}^{3}%
}e^{i\eta\cdot x+\left(  -ip\cdot\eta/m_{A}+L_{AB}\right)  t}\hat{g_{in}%
}\left(  \eta,p\right)  d\eta,
\]
where $\mathbb{G_{AB}}^{t}$ is the solution operator and $\hat{g}$ means the
Fourier transform in the space variable $x$. Furthermore, we introduce a long
wave-short wave decomposition for the solution $g$:%
\[
g=g_{L}+g_{S},
\]
where%
\[
g_{L}=\int_{\left\vert \eta\right\vert <\de}e^{i\eta\cdot x+\left(
-ip\cdot\eta/m_{A}+L_{AB}\right)  t}\hat{g_{in}}\left(  \eta,p\right)  d\eta,
\]%
\[
g_{S}=\int_{\left\vert \eta\right\vert \geq\de}e^{i\eta\cdot x+\left(
-ip\cdot\eta/m_{A}+L_{AB}\right)  t}\hat{g_{in}}\left(  \eta,p\right)  d\eta.
\]
In order to study the long wave part $g_{L},$ we need further to decompose it
into the fluid part and non-fluid part, i.e., $g_{L}=g_{L,0}+g_{L,\bot},$ where%

\[
g_{L,0}=\int_{|\eta|<\de}e^{\lambda(\eta)t}e^{i\eta\cdot x}\big<e_{D}%
(-\eta),\hat{g_{in}}\big>_{p}e_{D}(\eta)d\eta\,,
\]%
\[
g_{L,\bot}=\int_{|\eta|<\de}e^{i\eta\cdot x+\left(  -ip\cdot\eta/m_{A}%
+L_{AB}\right)  t}\Pi_{\eta}^{D\perp}\hat{g_{in}}\left(  \eta,p\right)
d\eta.
\]

The main result for the solution $g$ is stated as below.

\begin{theorem}
\label{theorem-g} Let $k\ $and$\ \ell\ $be nonnegative integers. Then for
$t\geq1,$
\[
\left\Vert \nabla_{p}^{\ell}\nabla_{x}^{k}g_{L,0}\right\Vert _{L_{x}^{\infty
}L_{p}^{2}}\lesssim(1+t)^{-(3+k)/2}\left\Vert g_{in}\right\Vert _{L_{x}%
^{1}L_{p}^{2}}\,,
\]%
\[
\left\Vert \nabla_{p}^{\ell}\nabla_{x}^{k}g_{L,\perp}\right\Vert
_{L_{x}^{\infty}L_{p}^{2}}\lesssim e^{-Ct}\left(  \left\Vert g_{in}\right\Vert
_{L^{2}\left(  w_{\ell+1}\right)  }+\left\Vert g_{in}\right\Vert _{L_{x}%
^{1}L_{p}^{2}}\right)  ,
\]
and
\[
\left\Vert \nabla_{p}^{\ell}\nabla_{x}^{k}g_{S}\right\Vert _{L_{x}^{\infty
}L_{p}^{2}}\lesssim e^{-Ct}\left\Vert g_{in}\right\Vert _{L^{2}\left(
w_{k+\ell+2}\right)  },
\]
the constant $C>0$ depending only upon $k$ and $\ell$. Here
\[
w_{n}\equiv%
\begin{cases}
1, & \ga\in\left[  0,1\right]  ,\\
\left\langle p\right\rangle ^{|\ga|n}, & \ga\in\left[  -2,0\right)  .
\end{cases}
\]

\end{theorem}

In what follows, we prove this theorem by several stages. To begin with, we
establish the $x$-regularity of $g$ in subsection \ref{improve-x}. After that,
we take advantage of it to improve the $p$-regularity via certain energy
functionals and Duhamel's principle in the next subsection. On the basis of
these regularization estimates, we finally obtain the decay rate by invoking
the Sobolev inequality and interpolation inequality.

\subsection{Improvement of the $x$-regularity}

\label{improve-x}

It immediately follows from the Fourier transform in $x$ and the
$p$-regularity of the eigenfunction $e_{D}(\eta)$ that
\[
\left\Vert \nabla_{p}^{\ell}\nabla_{x}^{k}g_{L,0}\right\Vert _{L_{x}^{\infty
}L_{p}^{2}}\lesssim(1+t)^{-(3+k)/2}\left\Vert g_{in}\right\Vert _{L_{x}%
^{1}L_{p}^{2}}.
\]
On the other hand, from the spectral gap it follows that
\[
\left\Vert \nabla_{x}^{k}g_{L,\perp}\right\Vert _{L_{x}^{\infty}L_{p}^{2}%
}\lesssim e^{-Ct}\left\Vert g_{in}\right\Vert _{L_{x}^{1}L_{p}^{2}},
\]
and
\[
\left\Vert g_{S}\right\Vert _{L^{2}}\lesssim e^{-Ct}\left\Vert g_{in}%
\right\Vert _{L^{2}}\,.
\]
Noticing that $g_{S}=g-g_{L},$ we need the regularization estimate of $g$ in
$x$ further. To that end, we introduce a new operator $\mathcal{L}_{AB}$:%
\[
\mathcal{L}_{AB}=-\frac{1}{m_{A}}p\cdot\nabla_{x}-\Lambda^{AB}.
\]
Let $n$ be a nonnegative integer and let $\theta\geq0$. Hereafter, we define
$m_{0}=m_{\theta}=\left\langle p\right\rangle ^{\theta}\ $and
\[
m_{n}\equiv%
\begin{cases}
m_{\theta}, & \ga\in\left[  0,1\right]  ,\\
\left\langle p\right\rangle ^{|\ga|n}m_{\theta}, & \ga\in\left[  -2,0\right)
.
\end{cases}
\]
Recall the energy estimate and the regularization property of the operator
$e^{t\mathcal{L}_{AB}}$ in \cite{[10]}:

\begin{lemma}
\label{regularization} If $u$ solves the equation
\begin{equation}
\left\{
\begin{array}
[c]{l}%
\pa_{t}u=\mathcal{L}_{AB}u\,,\\[4mm]%
u(0,x,p)=u_{0}(x,p)\,,
\end{array}
\right.
\end{equation}
then there exists $C>0$ such that
\[
\left\Vert e^{t\mathcal{L}_{AB}}u_{0}\right\Vert _{L^{2}(m_{\theta})}\leq
e^{-Ct}\left\Vert u_{0}\right\Vert _{L^{2}(m_{\theta})}\,.
\]
Moreover, for $0<t\leq1$, we have
\[
\int|\nabla_{p}e^{t\mathcal{L}_{AB}}u_{0}|^{2}m_{\theta}dxdp=O(t^{-1}%
)\int|u_{0}|^{2}m_{1}dxdp,
\]
and
\[
\int|\nabla_{x}e^{t\mathcal{L}_{AB}}u_{0}|^{2}m_{\theta}dxdp=O(t^{-3}%
)\int|u_{0}|^{2}m_{1}dxdp\,.
\]

\end{lemma}

Utilizing this lemma, we establish the regularity in $x$ as below. Worthy to be mentioned here, the proof for the case $k=1$ is
crucial. For the $\mathbf{x}$-derivatives, from Lemma
\ref{regularization}, $t^{-3/2}$ is not integrable as $t$ is
small. It appears to be harmful whenever we use the Duhamel principle. In
effect, one can see that the $L_{x}^{2}L_{p}^{2}\left(  m_{\theta}\right)
$ norms of the integrands in $\left(  \ref{aa.9}\right)
$ and $\left(  \ref{aa.10}\right)  $ are integrable in $t$
if we appropriately couple the operator $t\nabla_{x}$ with
$e^{t\mathcal{L}_{AB}}$.

\begin{lemma}
\label{x-regularityST} Let $f$ be a solution of the linearized Landau
equation
\begin{equation}
\left\{
\begin{array}
[c]{l}%
\partial_{t}f+\frac{1}{m_{A}}p\cdot\nabla_{x}f=L_{AB}f,%
\vspace{3mm}%
\\
f\left(  0,x,p\right)  =f_{in}\left(  x,p\right)  .
\end{array}
\right.  \label{Equ-AB}%
\end{equation}

Let $k\in\mathbb{N\cup\{}0\mathbb{\}}$. Then

\noindent\textrm{(i)}\quad\ for $0<t\leq1$%
\[
\left\Vert \nabla_{x}^{k}f\right\Vert _{L^{2}\left(  m_{\theta}\right)
}\lesssim t^{-\frac{3}{2}k}\left\Vert f_{in}\right\Vert _{L^{2}\left(
m_{k}\right)  }.
\]
\noindent\textrm{(ii)}\quad For $t\geq1,$%
\[
\left\Vert \nabla_{x}^{k}f\right\Vert _{L^{2}\left(  m_{\theta}\right)
}\lesssim\left\Vert f_{in}\right\Vert _{L^{2}\left(  m_{k}\right)  }.
\]

\end{lemma}

\begin{proof}
For $k=0,$ from Lemma \ref{energy-Guo} and the fact that $\left\Vert
f\right\Vert _{L^{2}}\leq\left\Vert f_{in}\right\Vert _{L^{2}}$, it follows
that%
\[
\frac{d}{dt}\left\Vert f\right\Vert _{L^{2}\left(  m_{\theta}\right)  }%
^{2}\lesssim-\left\Vert f\right\Vert _{L_{\sigma}^{2}\left(  m_{\theta
}\right)  }^{2}+c\left\Vert f\right\Vert _{L^{2}}^{2}\lesssim-\left\Vert
f\right\Vert _{L^{2}\left(  m_{\theta}\right)  }^{2}+c\left\Vert
f_{in}\right\Vert _{L^{2}}^{2},
\]
which implies that there exists a universal constant $C>0$ such that
\begin{equation}
\left\Vert f\right\Vert _{L^{2}\left(  m_{\theta}\right)  }\leq C\left\Vert
f_{in}\right\Vert _{L^{2}\left(  m_{\theta}\right)  }.
\label{f-weighted-energy}%
\end{equation}
Moreover, since the operator $\partial_{x}^{\alpha}$ is commutative with the
equation $\left(  \ref{Equ-AB}\right)  $,
\begin{equation}
\left\Vert \partial_{x}^{\alpha}f\left(  s_{2}\right)  \right\Vert
_{L^{2}\left(  m_{\theta}\right)  }\leq C\left\Vert \partial_{x}^{\alpha
}f\left(  s_{1}\right)  \right\Vert _{L^{2}\left(  m_{\theta}\right)
},\ \ \ 0<s_{1}<s_{2}, \label{fx-weighted-energy}%
\end{equation}
as well.

To prove this lemma for $k\geq1$, we design a Picard type iteration, treating
$K^{AB}$ as a source term, as below: The zero order approximation $f^{(0)}$ is
defined by
\begin{equation}
\left\{
\begin{array}
[c]{l}%
\pa_{t}f^{(0)}=\mathcal{L}_{AB}f^{(0)}\,,\\[4mm]%
f^{(0)}(0,x,p)=f_{in}(x,p)\,,
\end{array}
\right.  \label{bot.3.b}%
\end{equation}
and thus the difference $f-f^{(0)}$ satisfies
\[
\left\{
\begin{array}
[c]{l}%
\pa_{t}(f-f^{(0)})=\mathcal{L}_{AB}(f-f^{(0)})+K^{AB}(f-f^{(0)})+K^{AB}%
f^{(0)}\,,\\[4mm]%
(f-f^{(0)})(0,x,p)=0\,.
\end{array}
\right.
\]
Therefore, we define the first order approximation $f^{(1)}$ by
\begin{equation}
\left\{
\begin{array}
[c]{l}%
\pa_{t}f^{(1)}=\mathcal{L}_{AB}f^{(1)}+K^{AB}f^{(0)}\,,\\[4mm]%
f^{(1)}(0,x,p)=0\,.
\end{array}
\right.  \label{bot.3.c}%
\end{equation}
Inductively, we can define the $j^{\mathrm{th}}$ order approximation $f^{(j)}%
$, $j\geq1$, as
\begin{equation}
\left\{
\begin{array}
[c]{l}%
\pa_{t}f^{(j)}=\mathcal{L}_{AB}f^{(j)}+K^{AB}f^{(j-1)}\,,\\[4mm]%
f^{(j)}(0,x,p)=0\,.
\end{array}
\right.  \label{bot.3.d}%
\end{equation}
Now, rewrite $f$ as
\[
f=f^{(0)}+f^{(1)}+f^{(2)}+\cdot\cdot\cdot+f^{(2k)}+\mathcal{R}^{(k)}\,,
\]
here the remainder $\mathcal{R}^{(k)}$ solving the equation
\[
\left\{
\begin{array}
[c]{l}%
\partial_{t}\mathcal{R}^{(k)}+\frac{1}{m_{A}}p\cdot\nabla_{x}\mathcal{R}%
^{(k)}=L_{AB}\mathcal{R}^{(k)}+K^{AB}f^{(2k)},%
\vspace{3mm}%
\\
\mathcal{R}^{(k)}\left(  0,x,p\right)  =0\,.
\end{array}
\right.
\]
Under this decomposition,\textbf{ }we shall tackle the $x$-regularity for each
component in the cases\textbf{ }of\textbf{ }$k=1$ and $k=2.$ The case of
$k\geq3$ follows the same argument as $k=2.$\newline\textbf{Case 1}: $k=1$,
$f=\sum_{j=0}^{2}f^{(j)}+\mathcal{R}^{(1)}$. \newline Step 1: The first
$x$-derivative in small time. We claim that for $0<t\leq1$ and $0\leq j\leq
2$,
\[
\Vert\nabla_{x}f^{(j)}\Vert_{L^{2}(m_{\theta})}\lesssim t^{-\frac{3}{2}%
+j}\Vert f_{in}\Vert_{L^{2}(m_{1})},
\]
and
\[
\Vert\nabla_{x}\mathcal{R}^{(1)}\Vert_{L^{2}(m_{\theta})}\lesssim\Vert
f_{in}\Vert_{L^{2}(m_{1})}\,.
\]
\newline The estimate of $f^{(0)}$ is obviously from Lemma
\ref{regularization}. Notice that
\[
f^{(1)}=\int_{0}^{t}e^{(t-s)\mathcal{L}_{AB}}K^{AB}e^{s\mathcal{L}_{AB}}%
f_{in}ds\,.
\]
Hence,
\begin{align}
\nabla_{x}f^{(1)}  &  =\int_{0}^{t}\frac{(t-s)+s}{t}\nabla_{x}%
e^{(t-s)\mathcal{L}_{AB}}K^{AB}e^{s\mathcal{L}_{AB}}f_{in}ds\,\label{aa.9}\\
&  =\int_{0}^{t}\frac{1}{t}(t-s)\nabla_{x}e^{(t-s)\mathcal{L}_{AB}}%
K^{AB}e^{s\mathcal{L}_{AB}}f_{in}ds+\int_{0}^{t}\frac{1}{t}e^{(t-s)\mathcal{L}%
_{AB}}K^{AB}\left(  s\nabla_{x}e^{s\mathcal{L}_{AB}}f_{in}\right)
ds.\nonumber
\end{align}
By Lemma \ref{regularization},
\begin{align}
\left\Vert \nabla_{x}f^{(1)}\right\Vert _{L^{2}(m_{\theta})}  &  \lesssim
\int_{0}^{t}t^{-1}\left[  (t-s)^{-1/2}+s^{-1/2}\right]  ds\Vert f_{in}%
\Vert_{L^{2}(m_{1})}\label{f1-x}\\
&  \lesssim t^{-1/2}\Vert f_{in}\Vert_{L^{2}(m_{1})}\,.\nonumber
\end{align}
Likewise,
\[
f^{(2)}=\int_{0}^{t}\int_{0}^{s_{1}}e^{(t-s_{1})\mathcal{L}_{AB}}%
K^{AB}e^{(s_{1}-s_{2})\mathcal{L}_{AB}}K^{AB}e^{s_{2}\mathcal{L}_{AB}}%
f_{in}ds_{2}ds_{1}\,,
\]
and
\begin{align}
\nabla_{x}f^{(2)}  &  =\int_{0}^{t}\int_{0}^{s_{1}}\frac{(s_{1}-s_{2})+s_{2}%
}{s_{1}}\nabla_{x}e^{(t-s_{1})\mathcal{L}_{AB}}K^{AB}e^{(s_{1}-s_{2}%
)\mathcal{L}_{AB}}K^{AB}e^{s_{2}\mathcal{L}_{AB}}f_{in}ds_{2}ds_{1}%
\,\label{aa.10}\\
&  =\int_{0}^{t}\int_{0}^{s_{1}}\frac{1}{s_{1}}e^{(t-s_{1})\mathcal{L}_{AB}%
}K^{AB}\left[  (s_{1}-s_{2})\nabla_{x}e^{(s_{1}-s_{2})\mathcal{L}_{AB}}%
K^{AB}e^{s_{2}\mathcal{L}_{AB}}f_{in}\right]  ds_{2}ds_{1}\,\nonumber\\
&  +\int_{0}^{t}\int_{0}^{s_{1}}\frac{1}{s_{1}}e^{(t-s_{1})\mathcal{L}_{AB}%
}K^{AB}e^{(s_{1}-s_{2})\mathcal{L}_{AB}}K^{AB}\left[  s_{2}\nabla_{x}%
e^{s_{2}\mathcal{L}_{AB}}f_{in}\right]  ds_{2}ds_{1}\,.\nonumber
\end{align}
By Lemma \ref{regularization} again,
\begin{align}
\left\Vert \nabla_{x}f^{(2)}\right\Vert _{L^{2}(m_{\theta})}  &  \leq\int
_{0}^{t}\int_{0}^{s_{1}}s_{1}^{-1}\left[  (s_{1}-s_{2})^{-1/2}+s_{2}%
^{-1/2}\right]  ds_{2}ds_{1}\Vert f_{in}\Vert_{L^{2}(m_{1})}\label{aa.1}\\
&  \lesssim t^{1/2}\Vert f_{in}\Vert_{L^{2}(m_{1})}\,.\nonumber
\end{align}
Here we remark that $\nabla_{x}f^{(2)}$ is regular when $t$ is close to $0$.

For the remainder $\mathcal{R}^{(1)}$, since
\[
\nabla_{x}\mathcal{R}^{(1)}=\int_{0}^{t}\mathbb{G_{AB}}^{t-s}K^{AB}\nabla
_{x}f^{(2)}(s)ds\,,
\]
we deduce
\[
\left\Vert \nabla_{x}\mathcal{R}^{(1)}\right\Vert _{L^{2}(m_{\theta})}%
\lesssim\int_{0}^{t}\left\Vert \nabla_{x}f^{(2)}(s)\right\Vert _{L^{2}%
(m_{\theta})}ds\lesssim t^{3/2}\left\Vert f_{in}\right\Vert _{L^{2}(m_{1}%
)}\,,
\]
by $\left(  \ref{f-weighted-energy}\right)  $ and $\left(  \ref{aa.1}\right)
$. In conclusion, for $0<t\leq1$,
\begin{equation}
\left\Vert \nabla_{x}f\right\Vert _{L^{2}(m_{\theta})}\lesssim t^{-3/2}%
\left\Vert f_{in}\right\Vert _{L^{2}(m_{1})}\,. \label{f-x-stime}%
\end{equation}
\newline Step 2: The first $x$-derivative in large time. It readily follows
from $\left(  \ref{fx-weighted-energy}\right)  $ and $\left(  \ref{f-x-stime}%
\right)  $ that for $t>1$,
\[
\Vert\nabla_{x}f\left(  t\right)  \Vert_{L^{2}(m_{\theta})}\lesssim\Vert
\nabla_{x}f\left(  1\right)  \Vert_{L^{2}(m_{\theta})}\lesssim\left\Vert
f_{in}\right\Vert _{L^{2}(m_{1})}\,.
\]
\newline\textbf{Case 2}: $k=2$, $f=\sum_{j=0}^{4}f^{(j)}+\mathcal{R}^{(2)}$.
\newline Step 1: The second $x$-derivative in small time. Let $0<t\leq1$. We
shall show that
\begin{equation}
\Vert\nabla_{x}^{2}f^{(j)}\Vert_{L^{2}(m_{\theta})}\lesssim t^{-3+j}\Vert
f_{in}\Vert_{L^{2}(m_{2})},\ \ 0\leq j\leq4, \label{fj-xx}%
\end{equation}
and
\[
\Vert\nabla_{x}^{2}\mathcal{R}^{(2)}\Vert_{L^{2}(m_{\theta})}\lesssim
t^{2}\Vert f_{in}\Vert_{L^{2}(m_{2})}\,.
\]
For the first term, we only give the estimates of $f^{(0)}$ and $f^{(1)}$ and
the others are similar.

Let $0<t_{0}\leq1$ and $t_{0}/2<t\leq t_{0}$. Since
\[
\nabla_{x}f^{(0)}(t)=e^{(t-t_{0}/2)\mathcal{L}_{AB}}\nabla_{x}f^{(0)}%
(t_{0}/2)\,,
\]
we have
\[
\left\Vert \nabla_{x}^{2}f^{(0)}(t)\right\Vert _{L^{2}(m_{\theta})}%
\lesssim(t-t_{0}/2)^{-3/2}(t_{0}/2)^{-3/2}\Vert f_{in}\Vert_{L^{2}(m_{2})}\,,
\]
by Lemma \ref{regularization}. Taking $t=t_{0}$ gives
\begin{equation}
\left\Vert \nabla_{x}^{2}f^{(0)}(t_{0})\right\Vert _{L^{2}(m_{\theta}%
)}\lesssim t_{0}^{-3}\Vert f_{in}\Vert_{L^{2}(m_{2})}\,. \label{f0-xx}%
\end{equation}
Since $t_{0}\in(0,1]$ is arbitrary, this completes the proof of $f^{(0)}$.
Similarly, we have
\[
\nabla_{x}f^{(1)}(t)=e^{(t-t_{0}/2)\mathcal{L}_{AB}}\nabla_{x}f^{(1)}%
(t_{0}/2)+\int_{t_{0}/2}^{t}e^{(t-s)\mathcal{L}_{AB}}K^{AB}\nabla_{x}%
f^{(0)}(s)ds\,,
\]
\textbf{ }and thus it follows from Lemma \ref{regularization}, $\left(
\ref{f1-x}\right)  $ and $\left(  \ref{f0-xx}\right)  $ that
\[
\left\Vert \nabla_{x}^{2}f^{(1)}(t)\right\Vert _{L^{2}(m_{\theta})}%
\lesssim(t-t_{0}/2)^{-3/2}(t_{0}/2)^{-1/2}\Vert f_{in}\Vert_{L^{2}(m_{2}%
)}+\int_{t_{0}/2}^{t}s^{-3}\Vert f_{in}\Vert_{L^{2}(m_{2})}ds\,.
\]
Plugging $t=t_{0}$ into the above inequality yields
\begin{equation}
\left\Vert \nabla_{x}^{2}f^{(1)}(t_{0})\right\Vert _{L^{2}(m_{\theta}%
)}\lesssim t_{0}^{-2}\Vert f_{in}\Vert_{L^{2}(m_{2})},
\end{equation}
as desired.

For the remainder $\mathcal{R}^{(2)}$, since
\[
\nabla_{x}^{2}\mathcal{R}^{(2)}=\int_{0}^{t}\mathbb{G_{AB}}^{t-s}K^{AB}%
\nabla_{x}^{2}f^{(4)}(s)ds\,,
\]
we deduce
\[
\left\Vert \nabla_{x}^{2}\mathcal{R}^{(2)}\right\Vert _{L^{2}(m_{\theta}%
)}\lesssim\int_{0}^{t}\left\Vert \nabla_{x}^{2}f^{(4)}(s)\right\Vert
_{L^{2}(m_{\theta})}ds\lesssim t^{2}\left\Vert f_{in}\right\Vert _{L^{2}%
(m_{2})}\,,
\]
by $\left(  \ref{f-weighted-energy}\right)  $ and $\left(  \ref{fj-xx}\right)
.$ Consequently, for $0<t\leq1,$
\[
\left\Vert \nabla_{x}^{2}f\right\Vert _{L^{2}(m_{\theta})}\lesssim
t^{-3}\left\Vert f_{in}\right\Vert _{L^{2}(m_{2})}\,.\newline%
\]
Step 2: The second $x$-derivative in large time. For $t>1$, we have
\[
\Vert\nabla_{x}^{2}f\Vert_{L^{2}(m_{\theta})}\lesssim\Vert\nabla_{x}%
^{2}f\left(  1\right)  \Vert_{L^{2}(m_{\theta})}\lesssim\Vert f_{in}%
\Vert_{L^{2}(m_{2})}\,,
\]
due to $\left(  \ref{fx-weighted-energy}\right)  .$
\end{proof}

We now return to the $x$-regularity problem for $g_{S}.$ By Lemma
\ref{x-regularityST} $\left(  ii\right)  $ with initial condition
$f_{in}(x,p)=g_{in}(x,p)$ and taking $m_{\theta}=1,$
\[
\left\Vert \nabla_{x}^{k}g\right\Vert _{L^{2}}\lesssim\left\Vert
g_{in}\right\Vert _{L^{2}\left(  w_{k}\right)  },\text{ \ }k\in\mathbb{N\cup
}\left\{  0\right\}  ,
\]
for $t\geq1.$ Therefore,
\[
\left\Vert \nabla_{x}^{k}g_{S}\right\Vert _{L^{2}}\leq\left\Vert \nabla
_{x}^{k}g\right\Vert _{L^{2}}+\left\Vert \nabla_{x}^{k}g_{L}\right\Vert
_{L^{2}}\lesssim\left\Vert g_{in}\right\Vert _{L^{2}\left(  w_{k}\right)  },
\]
for each $k\in\mathbb{N\cup}\left\{  0\right\}  $. On the other hand, the
interpolation inequality says that%
\[
\left\Vert \nabla_{x}^{k}g_{S}\right\Vert _{L^{2}}\lesssim\left\Vert
\nabla_{x}^{k-1}g_{S}\right\Vert _{L^{2}}^{1/2}\left\Vert \nabla_{x}%
^{k+1}g_{S}\right\Vert _{L^{2}}^{1/2},\ \ \ k\in\mathbb{N}.
\]
Hence,
\begin{align}
\left\Vert \nabla_{x}^{k}g_{S}\right\Vert _{L^{2}}  &  \leq\left\{
\begin{array}
[c]{ll}%
\left\Vert g_{S}\right\Vert _{L^{2}}^{1/2}\left\Vert \nabla_{x}^{2}%
g_{S}\right\Vert _{L^{2}}^{1/2},%
\vspace{3mm}%
& k=1\\
\left\Vert g_{S}\right\Vert _{L^{2}}^{2^{-k}}\left(
{\displaystyle\prod\limits_{j=2}^{k}}
\left\Vert \nabla_{x}^{j}g_{S}\right\Vert _{L^{2}}^{2^{-(k-j+2)}}\right)
\left\Vert \nabla_{x}^{k+1}g_{S}\right\Vert _{L^{2}}^{1/2}, & k\geq2
\end{array}
\right.
\vspace{3mm}%
\label{aa.8}\\
&  \lesssim\left\Vert g_{S}\right\Vert _{L^{2}}^{2^{-k}}\left\Vert
g_{in}\right\Vert _{L^{2}\left(  w_{k+1}\right)  }^{1-2^{-k}}\nonumber\\
&  \lesssim e^{-2^{-k}Ct}\left\Vert g_{in}\right\Vert _{L^{2}(w_{k+1})}\,.%
\vspace{3mm}%
\nonumber
\end{align}
Combining this with the Sobolev inequality \cite[Proposition 3.8]{[Taylor]},
yields%
\begin{align*}
\left\Vert \nabla_{x}^{k}g_{S}\right\Vert _{L_{x}^{\infty}L_{p}^{2}}  &
\leq\left\Vert \nabla_{x}^{k}g_{S}\right\Vert _{L_{p}^{2}L_{x}^{\infty}%
}\lesssim\sqrt{\left\Vert \nabla_{x}^{k+1}g_{S}\right\Vert _{L^{2}}\left\Vert
\nabla_{x}^{k+2}g_{S}\right\Vert _{L^{2}}}\\
&  \lesssim\left(  e^{-2^{-k-1}Ct}\left\Vert g_{in}\right\Vert _{L^{2}%
(w_{k+2})}\right)  ^{1/2}\left\Vert g_{in}\right\Vert _{L^{2}\left(
w_{k+2}\right)  }^{1/2}\\
&  \lesssim e^{-2^{-k-2}Ct}\left\Vert g_{in}\right\Vert _{L^{2}(w_{k+2})}.
\end{align*}

Summing up, we conclude the following proposition.

\begin{proposition}
\label{x-regularity}Let $k,$ $\ell\in\mathbb{N\cup\{}0\}.$ Then for $t\geq1,$
\[
\left\Vert \nabla_{p}^{\ell}\nabla_{x}^{k}g_{L,0}\right\Vert _{L_{x}^{\infty
}L_{p}^{2}}\lesssim(1+t)^{-(3+k)/2}\left\Vert g_{in}\right\Vert _{L_{x}%
^{1}L_{p}^{2}},
\]%
\[
\left\Vert \nabla_{x}^{k}g_{L,\perp}\right\Vert _{L_{x}^{\infty}L_{p}^{2}%
}\lesssim e^{-Ct}\left\Vert g_{in}\right\Vert _{L_{x}^{1}L_{p}^{2}},
\]
and%
\[
\left\Vert \nabla_{x}^{k}g_{S}\right\Vert _{L_{x}^{\infty}L_{p}^{2}}\leq
e^{-C_{k}t}\left\Vert g_{in}\right\Vert _{L^{2}\left(  w_{k+2}\right)  },
\]
where $C>0,\ C_{k}>0.$ Here
\[
w_{n}\equiv%
\begin{cases}
1, & \ga\in\left[  0,1\right]  ,\\
\left\langle p\right\rangle ^{|\ga|n}, & \ga\in\left[  -2,0\right)  .
\end{cases}
\]

\end{proposition}

\subsection{Improvement of the $p$-regularity (general discussion)}

\label{improve-p}

In this subsection, we will study the $p$-regularity of the linearized Landau
equation
\begin{equation}
\left\{
\begin{array}
[c]{l}%
\displaystyle\partial_{t}f+\frac{1}{m_{A}}p\cdot\nabla_{x}f=L_{AB}f\,,\\[3mm]%
\displaystyle f\left(  0,x,p\right)  =f_{in}\left(  x,p\right)  .
\end{array}
\right.  \label{Landau-f}%
\end{equation}

\subsubsection{Improvement of the $p$-regularity in finite time}

\label{sub-p-regularity}

Assume that $0<t\leq1$. In order to improve the $p$-regularity of $f$ in
finite time, we first establish the $p$-regularity of the operator
$e^{t\mathcal{L}_{AB}}.$ Specifically, let $u$ be a solution of the equation
$\partial_{t}u=\mathcal{L}_{AB}u,$ $u(0,x,p)=u_{0}(x,p),$ and then we prove
that for each $\ell\in\mathbb{N}$, there exists a functional $\mathcal{F}%
_{\ell}$ such that
\[
t\left\Vert \nabla_{p}^{\ell}u\right\Vert ^{2}_{L^{2}(m_{\theta})}%
\lesssim\mathcal{F}_{\ell}\left(  0,u\right)  ,\ \ 0\leq t\leq1,
\]
where $\mathcal{F}_{\ell}\left(  0,u\right)  $ involves only the other
derivatives of the initial data $u_{0}$ with differentiation order less or
equal to $\ell.$ With the aid of the operator $e^{t\mathcal{L}_{AB}}$ and the
Duhamel principle, we use a bootstrap process to get the regularization
estimate of $f$ in \textbf{Both} the space variable $x$ and momentum variable
$p$ simultaneously.

Before proceeding, we derive some useful inequalities to simplify tedious
computations regarding the evolutions we will consider later on. Besides,
these inequalities clue us in on the form of the desired functional
$\mathcal{F}_{\ell}$.

\begin{lemma}
\label{energy-evolution}Let $u$ be a solution to the equation $\partial
_{t}u=\mathcal{L}_{AB}u.$ Then for $k,\ell\in\mathbb{N\cup\{}0\mathbb{\}},$
there exist $c_{0}>0$ independent\ of $k$ and $\ell$, and $C>0$ such that%
\begin{equation}
\frac{d}{dt}\left\Vert \nabla_{x}^{k}u\right\Vert _{L^{2}\left(  m_{\theta
}\right)  }^{2}\leq-c_{0}\left\Vert \nabla_{x}^{k}u\right\Vert _{L_{\sigma
}^{2}\left(  m_{\theta}\right)  }^{2}, \label{x-evolu}%
\end{equation}%
\begin{align}
\frac{d}{dt}\left\Vert \nabla_{p}^{\ell}u\right\Vert _{L^{2}\left(  m_{\theta
}\right)  }^{2}  &  \leq-c_{0}\left\Vert \nabla_{p}^{\ell}u\right\Vert
_{L_{\sigma}^{2}\left(  m_{\theta}\right)  }^{2}\label{p-evolu}\\
&  +C\left(  \sum_{j=0}^{\ell-1}\left\Vert \nabla_{p}^{j}u\right\Vert
_{L_{\sigma}^{2}\left(  m_{\theta}\right)  }^{2}+\sum_{j=0}^{\ell-1}\left\Vert
\nabla_{p}^{j}\nabla_{x}u\right\Vert _{L^{2}\left(  m_{\theta}\right)  }%
^{2}+\left\Vert u\right\Vert _{L^{2}}^{2}\right)  ,\nonumber
\end{align}
and%
\begin{align}
\frac{d}{dt}\left\Vert \nabla_{p}^{\ell}\nabla_{x}^{k}u\right\Vert
_{L^{2}\left(  m_{\theta}\right)  }^{2}  &  \leq-c_{0}\left\Vert \nabla
_{p}^{\ell}\nabla_{x}^{k}u\right\Vert _{L_{\sigma}^{2}\left(  m_{\theta
}\right)  }^{2}\label{mixed-evolu}\\
&  +C\left(  \sum_{j=0}^{\ell-1}\left\Vert \nabla_{p}^{j}\nabla_{x}%
^{k}u\right\Vert _{L_{\sigma}^{2}\left(  m_{\theta}\right)  }^{2}+\sum
_{j=0}^{\ell-1}\left\Vert \nabla_{p}^{j}\nabla_{x}^{k+1}u\right\Vert
_{L^{2}\left(  m_{\theta}\right)  }^{2}+\left\Vert \nabla_{x}^{k}u\right\Vert
_{L^{2}}^{2}\right)  .\nonumber
\end{align}

\end{lemma}

\begin{proof}
Let $\alpha\ $and $\beta$ be any multi-indices. Direct computation shows
\begin{align*}
\frac{1}{2}\frac{d}{dt}\left\Vert \partial_{x}^{\alpha}u\right\Vert
_{L^{2}\left(  m_{\theta}\right)  }^{2}  &  =\frac{1}{2}\int\nabla_{x}%
\cdot\left(  \mathcal{-}\frac{p}{m_{A}}\left\vert \partial_{x}^{\alpha
}u\right\vert ^{2}\right)  m_{\theta}dxdp\\
&  +\int\left(  \partial_{x}^{\alpha}u\nabla_{p}\cdot\left[  \sigma^{AB}%
\nabla_{p}\partial_{x}^{\alpha}u\right]  -\varphi\left(  p\right)  \left\vert
\partial_{x}^{\alpha}u\right\vert ^{2}\right)  m_{\theta}dxdp\\
&  =-\int\left(  \sigma^{AB}\nabla_{p}\partial_{x}^{\alpha}u,\nabla
_{p}\partial_{x}^{\alpha}u\right)  m_{\theta}dxdp-\int\left[  \varphi\left(
p\right)  \left\vert \partial_{x}^{\alpha}u\right\vert ^{2}m_{\theta}+\frac
{1}{2}\left(  \sigma^{AB}\nabla_{p}\left\vert \partial_{x}^{\alpha
}u\right\vert ^{2},\nabla_{p}m_{\theta}\right)  \right]  dxdp\\
&  =-\int\left(  \sigma^{AB}\nabla_{p}\partial_{x}^{\alpha}u,\nabla
_{p}\partial_{x}^{\alpha}u\right)  m_{\theta}dxdp\\
&  -\int\left[  \varphi\left(  p\right)  -\frac{1}{2m_{\theta}}\sum
_{i,j=1}^{3}\left(  \partial_{j}\sigma_{ij}^{AB}\partial_{i}m_{\theta}%
+\sigma_{ij}^{AB}\partial_{ij}^{2}m_{\theta}\right)  \right]  \left\vert
\partial_{x}^{\alpha}u\right\vert ^{2}m_{\theta}dxdp,
\end{align*}
where $\partial_{j}=\partial_{p_{j}}$, $\partial_{ij}^{2}$ $=\partial
_{p_{i}p_{j}}^{2}$ and%

\[
\varphi\left(  p\right)  =\frac{1}{4}\frac{m_{B}^{2}}{m_{A}^{2}}\left(
\sigma^{AB}p,p\right)  -\frac{1}{2}\frac{m_{B}}{m_{A}}\nabla_{p}\cdot\left(
\sigma^{AB}p\right)  +\varpi\chi_{R}.
\]
By Lemma \ref{decom} $(iii)$,
\[
\left\vert \frac{1}{2m_{\theta}}\sum_{i,j=1}^{3}\left(  \partial_{j}%
\sigma_{ij}^{AB}\partial_{i}m_{\theta}+\sigma_{ij}^{AB}\partial_{ij}m_{\theta
}\right)  \right\vert \leq c\left\langle p\right\rangle ^{\gamma},
\]
which implies that
\begin{align*}
&  \varphi\left(  p\right)  -\frac{1}{2m_{\theta}}\sum_{i,j=1}^{3}\left(
\partial_{j}\sigma_{ij}^{AB}\partial_{i}m_{\theta}+\sigma_{ij}^{AB}%
\partial_{ij}m_{\theta}\right) \\
&  \geq\frac{1}{4}\lambda_{1}^{AB}\left(  p\right)  \left(  \frac{m_{B}}%
{m_{A}}\left\vert p\right\vert \right)  ^{2}+\frac{1}{2}\lambda_{1}%
^{AB}\left(  p\right)  \frac{m_{B}}{m_{A}}\left(  \frac{m_{B}}{m_{A}%
}\left\vert p\right\vert \right)  ^{2}-\frac{m_{B}}{m_{A}}\lambda_{2}%
^{AB}\left(  p\right)  -\frac{1}{2}\frac{m_{B}}{m_{A}}\lambda_{1}^{AB}\left(
p\right)  +\varpi\chi_{R}-c\left\langle p\right\rangle ^{\gamma}\\
&  \geq c_{1}\left\langle p\right\rangle ^{\gamma+2},
\end{align*}
$\ $provided $\varpi$ and $R$ are sufficiently large. Owing to Lemma
\ref{decom} $(ii)$, there exists $c_{2}>0$ independent\ of $\alpha$ such that
\[
\frac{d}{dt}\left\Vert \partial_{x}^{\alpha}u\right\Vert _{L^{2}\left(
m_{\theta}\right)  }^{2}\leq-c_{2}\left\Vert \partial_{x}^{\alpha}u\right\Vert
_{L_{\sigma}^{2}\left(  m_{\theta}\right)  }^{2},
\]
as required.

For $\left\vert \alpha\right\vert \geq1,$ compute the evolution
\begin{align*}
&  \frac{1}{2}\frac{d}{dt}\left\Vert \partial_{p}^{\alpha}u\right\Vert
_{L^{2}\left(  m_{\theta}\right)  }^{2}\\
&  =\int\partial_{p}^{\alpha}u\partial_{p}^{\alpha}\left(  -\frac{p}{m_{A}%
}\cdot\nabla_{x}u-\widetilde{\Lambda}^{AB}u-\varpi\chi_{R}u\right)  m_{\theta
}dxdp\\
&  =\frac{1}{2}\int\nabla_{x}\cdot\left(  -\frac{p}{m_{A}}\left\vert
\partial_{p}^{\alpha}u\right\vert ^{2}\right)  m_{\theta}dxdp-\int\left[
\sum_{\beta<\alpha,\ \left\vert \beta\right\vert =\left\vert \alpha\right\vert
-1}C_{\beta}^{\alpha}\partial_{p}^{\alpha-\beta}\left(  \frac{p}{m_{A}%
}\right)  \cdot\nabla_{x}\partial_{p}^{\beta}u\right]  \partial_{p}^{\alpha
}um_{\theta}dxdp\\
&  -\int\left(  \partial_{p}^{\alpha}\widetilde{\Lambda}^{AB}u\right)
\partial_{p}^{\alpha}um_{\theta}dxdp-\varpi\int\partial_{p}^{\alpha}%
u\partial_{p}^{\alpha}\left(  \chi_{R}u\right)  m_{\theta}dxdp\\
&  =-\int\left[  \sum_{\beta<\alpha,\ \left\vert \beta\right\vert =\left\vert
\alpha\right\vert -1}C_{\beta}^{\alpha}\partial_{p}^{\alpha-\beta}\left(
\frac{p}{m_{A}}\right)  \cdot\nabla_{x}\partial_{p}^{\beta}u\right]
\partial_{p}^{\alpha}um_{\theta}dxdp-\int\left(  \partial_{p}^{\alpha
}\widetilde{\Lambda}^{AB}u\right)  \partial_{p}^{\alpha}um_{\theta}dxdp\\
&  -\varpi\int\chi_{R}\left\vert \partial_{p}^{\alpha}u\right\vert
^{2}m_{\theta}dxdp-\varpi\int\sum_{\beta<\alpha}C_{\beta}^{\alpha}\partial
_{p}^{\alpha}u\partial_{p}^{\alpha-\beta}\chi_{R}\partial_{p}^{\beta
}um_{\theta}dxdp,
\end{align*}
here we denote $\binom{\alpha}{\beta}\ $by $C_{\beta}^{\alpha}.$ By choosing
$0<\delta\ll1$ and the Cauchy inequality,
\[
\left\vert \int\left[  \sum_{\substack{\beta<\alpha,\\\left\vert
\beta\right\vert =\left\vert \alpha\right\vert -1}}C_{\beta}^{\alpha}%
\partial_{p}^{\alpha-\beta}\left(  \frac{p}{m_{A}}\right) \cdot \nabla_{x}%
\partial_{p}^{\beta}u\right]  \partial_{p}^{\alpha}um_{\theta}dxdp\right\vert
\leq\delta\left\Vert \partial_{p}^{\alpha}u\right\Vert _{L^{2}\left(
m_{\theta}\right)  }^{2}+C_{\delta}\sum_{\left\vert \overline{\alpha
}\right\vert <\left\vert \alpha\right\vert }\left\Vert \nabla_{x}\partial
_{p}^{\overline{\alpha}}u\right\Vert _{L^{2}\left(  m_{\theta}\right)  }^{2},
\]
and%
\[
\left\vert \varpi\int\sum_{\beta<\alpha}C_{\beta}^{\alpha}\partial_{p}%
^{\alpha}u\partial_{p}^{\alpha-\beta}\chi_{R}\partial_{p}^{\beta}um_{\theta
}dxdp\right\vert \leq\delta\left\Vert \partial_{p}^{\alpha}u\right\Vert
_{L^{2}\left(  m_{\theta}\right)  }^{2}+C_{\delta}\sum_{\left\vert
\overline{\alpha}\right\vert <\left\vert \alpha\right\vert }\left\Vert
\partial_{p}^{\overline{\alpha}}u\right\Vert _{L^{2}\left(  m_{\theta}\right)
}^{2}.
\]
From Lemma \ref{energy-Guo},
\[
-\int\left(  \partial_{p}^{\alpha}\widetilde{\Lambda}^{AB}u\right)
\partial_{p}^{\alpha}um_{\theta}dxdp\leq-c_{1}\left\Vert \partial_{p}^{\alpha
}u\right\Vert _{L_{\sigma}^{2}\left(  m_{\theta}\right)  }^{2}+\delta
\sum_{\left\vert \overline{\alpha}\right\vert \leq\left\vert \alpha\right\vert
}\left\Vert \partial_{p}^{\overline{\alpha}}u\right\Vert _{L_{\sigma}%
^{2}\left(  m_{\theta}\right)  }^{2}+C_{\delta}\left\Vert \chi_{R}u\right\Vert
_{L^{2}}^{2}.
\]
Together with $\left\Vert \partial_{p}^{\alpha}u\right\Vert _{L^{2}\left(
m_{\theta}\right)  }^{2}\leq\left\Vert \partial_{p}^{\alpha}u\right\Vert
_{L_{\sigma}^{2}\left(  m_{\theta}\right)  }^{2},$ we therefore have
\[
\frac{d}{dt}\left\Vert \nabla_{p}^{\ell}u\right\Vert _{L^{2}\left(  m_{\theta
}\right)  }^{2}\leq-\frac{c_{1}}{2}\left\Vert \partial_{p}^{\ell}u\right\Vert
_{L_{\sigma}^{2}\left(  m_{\theta}\right)  }^{2}+C\left(  \sum_{j=0}^{\ell
-1}\left\Vert \nabla_{p}^{j}u\right\Vert _{L_{\sigma}^{2}\left(  m_{\theta
}\right)  }^{2}+\sum_{j=0}^{\ell-1}\left\Vert \nabla_{p}^{j}\nabla
_{x}u\right\Vert _{L^{2}\left(  m_{\theta}\right)  }^{2}+\left\Vert
u\right\Vert _{L^{2}}^{2}\right)  ,
\]
for $\ell\in\mathbb{N}$. Since $\partial_{x}^{\alpha}$ is commutative with
$\mathcal{L}_{AB},$ $\left(  \ref{mixed-evolu}\right)  $ is a consequence of
$\left(  \ref{p-evolu}\right)  $.
\end{proof}

Now, we embark on the $p$-regularity estimate for $e^{t\mathcal{L}_{AB}}$ and
$f$ in turn. For clarification, we first elaborate our procedure in the cases
of $\ell=1$ and $\ell=2.$ For general $\ell,$ we provide the explicit form of
the desired functional $\mathcal{F}_{\ell}$ and complete the proof inductively
on $\ell.$\newline Step 1: The estimate of $\nabla_{p}\nabla_{x}^{k}f$ (i.e.,
$\ell=1$). Define the functional
\[
\mathcal{F}_{1}(t,u)\equiv\int u^{2}m_{1}dxdp+\int\left\vert \nabla
_{x}u\right\vert ^{2}m_{\theta}dxdp+\kappa t\int\left\vert \nabla
_{p}u\right\vert ^{2}m_{\theta}dxdp.
\]
In view of Lemma \ref{energy-evolution},
\[
\frac{d}{dt}\int u^{2}m_{1}dxdp\lesssim-\Vert u\Vert_{L_{\sigma}^{2}(m_{1}%
)}^{2},\ \text{\ \ \ \ }\frac{d}{dt}\int\left\vert \nabla_{x}u\right\vert
^{2}m_{\theta}dxdp\lesssim-\Vert\nabla_{x}u\Vert_{L_{\sigma}^{2}(m_{\theta}%
)}^{2},
\]%
\[
\frac{d}{dt}\int\left\vert \nabla_{p}u\right\vert ^{2}m_{\theta}%
dxdp\lesssim-\Vert\nabla_{p}u\Vert_{L_{\sigma}^{2}(m_{\theta})}^{2}+C\left(
\Vert\nabla_{x}u\Vert_{L_{\sigma}^{2}(m_{\theta})}^{2}+\Vert u\Vert
_{L_{\sigma}^{2}(m_{\theta})}^{2}\right)  .
\]
Collecting terms gives
\begin{align*}
\frac{d}{dt}\mathcal{F}_{1}(t,u)  &  \lesssim\Vert u\Vert_{L_{\sigma}%
^{2}(m_{1})}^{2}\left(  -1+C\kappa t\right)  +C\kappa\int\left\vert \nabla
_{p}u\right\vert ^{2}m_{\theta}\\
&  +\Vert\nabla_{x}u\Vert_{L_{\sigma}^{2}(m_{\theta})}^{2}\left(  -1+C\kappa
t\right)  +\Vert\nabla_{p}u\Vert_{L_{\sigma}^{2}(m_{\theta})}^{2}\left(
-\kappa t\right)  .
\end{align*}
Let $\kappa=\varepsilon^{2}$. Choosing $\varepsilon>0\ $sufficiently small, we
thereby obtain
\[
\frac{d}{dt}\mathcal{F}_{1}(t,u)<0\ \ \ \ \text{for \ }0<t<1.
\]
It implies that for $0\leq t\leq1,$%
\begin{equation}
t\left\Vert \nabla_{p}u\left(  t\right)  \right\Vert _{L^{2}\left(  m_{\theta
}\right)  }^{2}\lesssim\mathcal{F}_{1}(0,u)=\left\Vert u\left(  0\right)
\right\Vert _{L^{2}\left(  m_{1}\right)  }^{2}+\left\Vert \nabla_{x}u\left(
0\right)  \right\Vert _{L^{2}\left(  m_{\theta}\right)  }^{2}\,. \label{aa.2}%
\end{equation}

Next, we return to the estimate for $f$. Let $0<t_{0}\leq1$ be given. By
Duhamel's principle,
\[
f\left(  t\right)  =e^{(t-t_{0}/2)\mathcal{L}_{AB}}f(t_{0}/2)+\int_{t_{0}%
/2}^{t}e^{\left(  t-s\right)  \mathcal{L}_{AB}}K^{AB}f\left(  s\right)  ds,
\]
for $0<t_{0}/2<t\leq1.$ Hence, we have
\begin{align}
\left\Vert \nabla_{p}f\left(  t\right)  \right\Vert _{L^{2}\left(  m_{\theta
}\right)  }  &  \leq C\left(  t-t_{0}/2\right)  ^{-\frac{1}{2}}\left(
\left\Vert f\left(  t_{0}/2\right)  \right\Vert _{L^{2}\left(  m_{1}\right)
}+\left\Vert \nabla_{x}f\left(  t_{0}/2\right)  \right\Vert _{L^{2}\left(
m_{\theta}\right)  }\right) \label{fp}\\
&  \hspace{1em}+C\int_{t_{0}/2}^{t}(t-s)^{-\frac{1}{2}}\left(  \left\Vert
f\left(  s\right)  \right\Vert _{L^{2}\left(  m_{1}\right)  }+\left\Vert
\nabla_{x}f\left(  s\right)  \right\Vert _{L^{2}\left(  m_{\theta}\right)
}\right)  ds\nonumber\\
&  \leq C\left(  t-t_{0}/2\right)  ^{-\frac{1}{2}}t_{0}^{-\frac{3}{2}%
}\left\Vert f_{in}\right\Vert _{L^{2}\left(  m_{1}\right)  }\,,\nonumber
\end{align}
due to Lemma \ref{x-regularityST}, $\left(  \ref{f-weighted-energy}\right)
\ $and $(\ref{aa.2}).$ Consequently,
\[
\left\Vert \nabla_{p}f\left(  t\right)  \right\Vert _{L^{2}\left(  m_{\theta
}\right)  }\lesssim t^{-\frac{1}{2}-\frac{3}{2}}\left\Vert f_{in}\right\Vert
_{L^{2}\left(  m_{1}\right)  },\ \ \ \ \ 0<t\leq1.
\]
On the other hand, since $\partial_{x}^{\alpha}$ commutes with the equation
$\left(  \ref{Landau-f}\right)  $, we can improve the mixed regularity
simultaneously. Precisely, replacing $f$ by $\partial_{x}^{\alpha}f$ in
$\left(  \ref{fp}\right)  $ gives
\begin{align*}
\left\Vert \nabla_{p}\nabla_{x}^{k}f\left(  t\right)  \right\Vert
_{L^{2}\left(  m_{\theta}\right)  }  &  \leq C\left(  t-t_{0}/2\right)
^{-\frac{1}{2}}\left(  \left\Vert \nabla_{x}^{k}f\left(  t_{0}/2\right)
\right\Vert _{L^{2}\left(  m_{1}\right)  }+\left\Vert \nabla_{x}^{k+1}f\left(
t_{0}/2\right)  \right\Vert _{L^{2}\left(  m_{\theta}\right)  }\right) \\
&  +C\int_{t_{0}/2}^{t}(t-s)^{-\frac{1}{2}}\left(  \left\Vert \nabla_{x}%
^{k}f\left(  s\right)  \right\Vert _{L^{2}\left(  m_{1}\right)  }+\left\Vert
\nabla_{x}^{k+1}f\left(  s\right)  \right\Vert _{L^{2}\left(  m_{\theta
}\right)  }\right)  ds.
\end{align*}
By Lemma \ref{x-regularityST} and $\left(  \ref{f-weighted-energy}\right)  ,$
\begin{equation}
\left\Vert \nabla_{p}\nabla_{x}^{k}f\left(  t\right)  \right\Vert
_{L^{2}\left(  m_{\theta}\right)  }\leq Ct^{-\frac{1}{2}-\frac{3}{2}%
(k+1)}\left\Vert f_{in}\right\Vert _{L^{2}\left(  m_{k+1}\right)  }.
\label{aa.5}%
\end{equation}
\newline Step 2: The estimate of $\nabla_{p}^{2}\nabla_{x}^{k}f\ $(i.e.,
$\ell=2$). Define the functional
\[
\mathcal{F}_{2}(t,u)\equiv\int(u^{2}+\left\vert \nabla_{x}u\right\vert
^{2}+\kappa_{1}\left\vert \nabla_{p}u\right\vert ^{2})m_{1}dxdp+\int
(\left\vert \nabla_{x}^{2}u\right\vert ^{2}+\kappa_{2}\left\vert \nabla
_{p}\nabla_{x}u\right\vert ^{2})m_{\theta}dxdp+\kappa_{3}t\int\left\vert
\nabla_{p}^{2}u\right\vert ^{2}m_{\theta}dxdp.
\]
By Lemma \ref{energy-evolution},
\[
\frac{d}{dt}\int(u^{2}+\left\vert \nabla_{x}u\right\vert ^{2})m_{1}%
dxdp\lesssim-\Vert u\Vert_{L_{\sigma}^{2}(m_{1})}^{2}-\Vert\nabla_{x}%
u\Vert_{L_{\sigma}^{2}(m_{1})}^{2}\,,
\]

\[
\frac{d}{dt}\int\left\vert \nabla_{p}u\right\vert ^{2}m_{1}dxdp\lesssim
-\Vert\nabla_{p}u\Vert_{L_{\sigma}^{2}(m_{1})}^{2}+C\left(  \Vert\nabla
_{x}u\Vert_{L_{\sigma}^{2}(m_{1})}^{2}+\Vert u\Vert_{L_{\sigma}^{2}(m_{1}%
)}^{2}\right)  \,,
\]

\[
\frac{d}{dt}\int\left\vert \nabla_{x}^{2}u\right\vert ^{2}m_{\theta
}dxdp\lesssim-\Vert\nabla_{x}^{2}u\Vert_{L_{\sigma}^{2}(m_{\theta})}^{2}\,,
\]

\[
\frac{d}{dt}\int\left\vert \nabla_{p}\nabla_{x}u\right\vert ^{2}m_{\theta
}dxdp\lesssim-\Vert\nabla_{p}\nabla_{x}u\Vert_{L_{\sigma}^{2}(m_{\theta})}%
^{2}+C\left(  \Vert\nabla_{x}^{2}u\Vert_{L_{\sigma}^{2}(m_{\theta})}^{2}%
+\Vert\nabla_{x}u\Vert_{L_{\sigma}^{2}(m_{\theta})}^{2}\right)  \,,
\]

\[
\frac{d}{dt}\int\left\vert \nabla_{p}^{2}u\right\vert ^{2}m_{\theta
}dxdp\lesssim-\Vert\nabla_{p}^{2}u\Vert_{L_{\sigma}^{2}(m_{\theta})}%
^{2}+C\left(  \Vert\nabla_{p}\nabla_{x}u\Vert_{L_{\sigma}^{2}(m_{\theta})}%
^{2}+\Vert u\Vert_{L_{\sigma}^{2}(m_{\theta})}^{2}+\Vert\nabla_{p}%
u\Vert_{L_{\sigma}^{2}(m_{\theta})}^{2}\right)  \,.
\]
Hence,
\begin{align*}
\frac{d}{dt}\mathcal{F}_{2}(t,u)  &  \lesssim\Vert u\Vert_{L_{\sigma}%
^{2}(m_{1})}^{2}\left(  -1+C\kappa_{1}+C\kappa_{3}t\right)  +\Vert\nabla
_{x}u\Vert_{L_{\sigma}^{2}(m_{1})}^{2}\left(  -1+C\kappa_{1}+C\kappa
_{2}\right) \\
&  +\Vert\nabla_{p}u\Vert_{L_{\sigma}^{2}(m_{1})}^{2}\left(  -\kappa
_{1}+C\kappa_{3}t\right)  +C\kappa_{3}\Vert\nabla_{p}^{2}u\Vert_{L^{2}%
(m_{\theta})}^{2}\\
&  +\Vert\nabla_{x}^{2}u\Vert_{L_{\sigma}^{2}(m_{\theta})}^{2}\left(
-1+C\kappa_{2}\right)  +\Vert\nabla_{p}\nabla_{x}u\Vert_{L_{\sigma}%
^{2}(m_{\theta})}^{2}\left(  -\kappa_{2}+C\kappa_{3}t\right) \\
&  +\Vert\nabla_{p}^{2}u\Vert_{L_{\sigma}^{2}(m_{\theta})}^{2}\left(
-\kappa_{3}t\right)  .
\end{align*}
Set $\kappa_{i}=\varepsilon^{i}$. Choosing $\varepsilon>0$ sufficiently small
, we find%
\[
\frac{d}{dt}\mathcal{F}_{2}(t,u)\leq0,\ \ \ \text{for \ }0<t<1,\
\]
so that
\begin{equation}
t\left\Vert \nabla_{p}^{2}u(t)\right\Vert _{L^{2}(m_{\theta})}^{2}%
\lesssim\mathcal{F}_{2}(0,u)=\left\Vert u(0)\right\Vert _{H^{1}(m_{1})}%
^{2}+\left\Vert \nabla_{x}^{2}u(0)\right\Vert _{L^{2}(m_{\theta})}%
^{2}+\left\Vert \nabla_{p}\nabla_{x}u(0)\right\Vert _{L^{2}(m_{\theta})}%
^{2}\,, \label{aa.4}%
\end{equation}
for $0\leq t\leq1.$

Let $0<t_{0}\leq1.$ By Duhamel's principle,
\[
f\left(  t\right)  =e^{(t-t_{0}/2)\mathcal{L}_{AB}}f(t_{0}/2)+\int_{t_{0}%
/2}^{t}e^{\left(  t-s\right)  \mathcal{L}_{AB}}K^{AB}f\left(  s\right)  ds,
\]
for $0<t_{0}/2<t\leq1.$ Using Lemma \ref{x-regularityST}, $\left(
\ref{f-weighted-energy}\right)  ,(\ref{aa.5})$ and $(\ref{aa.4})$, we obtain
\begin{align}
\left\Vert \nabla_{p}^{2}f\left(  t\right)  \right\Vert _{L^{2}\left(
m_{\theta}\right)  }  &  \leq C\left(  t-t_{0}/2\right)  ^{-\frac{1}{2}%
}\left(  \left\Vert f(t_{0}/2)\right\Vert _{H^{1}(m_{1})}+\left\Vert
\nabla_{x}^{2}f(t_{0}/2)\right\Vert _{L^{2}(m_{\theta})}+\left\Vert \nabla
_{p}\nabla_{x}f(t_{0}/2)\right\Vert _{L^{2}(m_{\theta})}\right) \label{fpp}\\
&  \hspace{1em}+C\int_{t_{0}/2}^{t}(t-s)^{-\frac{1}{2}}\left(  \left\Vert
f(s)\right\Vert _{H^{1}(m_{1})}+\left\Vert \nabla_{x}^{2}f(s)\right\Vert
_{L^{2}(m_{\theta})}+\left\Vert \nabla_{p}\nabla_{x}f(s)\right\Vert
_{L^{2}(m_{\theta})}\right)  ds\nonumber\\
&  \leq C\left(  t-t_{0}/2\right)  ^{-\frac{1}{2}}t_{0}^{-\frac{1}{2}-\frac
{3}{2}\times2}\left\Vert f_{in}\right\Vert _{L^{2}\left(  m_{2}\right)
}\,.\nonumber
\end{align}
Namely,
\[
\left\Vert \nabla_{p}^{2}f\left(  t\right)  \right\Vert _{L^{2}\left(
m_{\theta}\right)  }\lesssim t^{-\frac{2}{2}-\frac{3}{2}\times2}\left\Vert
f_{in}\right\Vert _{L^{2}\left(  m_{2}\right)  },\ \ \ \ \ 0<t\leq1.
\]
Again, since $\partial_{x}^{\alpha}$ is commutative with the equation $\left(
\ref{Landau-f}\right)  $, replacing $f$ by $\partial_{x}^{\alpha}f$ in
$\left(  \ref{fpp}\right)  $ gives
\begin{align*}
\left\Vert \nabla_{p}^{2}\nabla_{x}^{k}f\left(  t\right)  \right\Vert
_{L^{2}\left(  m_{\theta}\right)  }  &  \leq C\left(  t-t_{0}/2\right)
^{-\frac{1}{2}}\left(  \left\Vert \nabla_{x}^{k}f(t_{0}/2)\right\Vert
_{H^{1}(m_{1})}+\left\Vert \nabla_{x}^{k+2}f(t_{0}/2)\right\Vert
_{L^{2}(m_{\theta})}+\left\Vert \nabla_{p}\nabla_{x}^{k+1}f(t_{0}%
/2)\right\Vert _{L^{2}(m_{\theta})}\right) \\
&  +C\int_{t_{0}/2}^{t}(t-s)^{-\frac{1}{2}}\left(  \left\Vert \nabla_{x}%
^{k}f(s)\right\Vert _{H^{1}(m_{1})}+\left\Vert \nabla_{x}^{k+2}f(s)\right\Vert
_{L^{2}(m_{\theta})}+\left\Vert \nabla_{p}\nabla_{x}^{k+1}f(s)\right\Vert
_{L^{2}(m_{\theta})}\right)  ds\\
&  \leq C\left(  t-t_{0}/2\right)  ^{-\frac{1}{2}}t_{0}^{-\frac{1}{2}-\frac
{3}{2}\left(  k+2\right)  }\left\Vert f_{in}\right\Vert _{L^{2}\left(
m_{k+2}\right)  }.
\end{align*}
That is,
\[
\left\Vert \nabla_{p}^{2}\nabla_{x}^{k}f\left(  t\right)  \right\Vert
_{L^{2}\left(  m_{\theta}\right)  }\leq Ct^{-\frac{2}{2}-\frac{3}{2}%
(k+2)}\left\Vert f_{in}\right\Vert _{L^{2}\left(  m_{k+2}\right)  }.
\]
\newline Step 3: The estimate of $\nabla_{p}^{\ell}\nabla_{x}^{k}f$. For
general $\ell\in\mathbb{N}$, consider the functional
\begin{align*}
\mathcal{F}_{\ell}\left(  t,u\right)   &  =\int\left(  \sum_{j=0}^{\ell
-1}\left\vert \nabla_{x}^{j}u\right\vert ^{2}+\sum_{j=1}^{\ell-1}\sum
_{q=1}^{j}\varepsilon^{q}\left\vert \nabla_{p}^{q}\nabla_{x}^{j-q}u\right\vert
^{2}\right)  m_{1}dxdp+\int\left(  \left\vert \nabla_{x}^{\ell}u\right\vert
^{2}+\sum_{j=1}^{\ell-1}\varepsilon^{j+1}\left\vert \nabla_{p}^{j}\nabla
_{x}^{\ell-j}u\right\vert ^{2}\right)  m_{\theta}dxdp%
\vspace{3mm}%
\\
&  +\varepsilon^{\ell+1}t\int\left\vert \nabla_{p}^{\ell}u\right\vert
^{2}m_{\theta}dxdp,
\end{align*}
for $\varepsilon>0$ sufficiently small. Applying the above argument
inductively on $\ell$, we obtain high order $p$-regularity in small time and
conclude our result as below.

\begin{proposition}
\label{Mixed-decay}Let $k,$ $\ell\in\mathbb{N}\cup\{0\}.$ Then for
$0<t\leq1,$
\[
\left\Vert \nabla_{p}^{\ell}\nabla_{x}^{k}f\left(  t\right)  \right\Vert
_{L^{2}\left(  m_{\theta}\right)  }\leq Ct^{-\frac{\ell}{2}-\frac{3}{2}%
(k+\ell)}\left\Vert f_{in}\right\Vert _{L^{2}\left(  m_{k+\ell}\right)  \ }.
\]

\end{proposition}

\subsubsection{Improvement of the $p$ regularity in large time}

In what follows, we establish the $p$-regularity in large time through the
Gronwall type inequalities.

\begin{proposition}
\label{energy*} Let $f$ be a solution to equation $\left(  \ref{Landau-f}%
\right)  $ and let $k,\ \ell\in\mathbb{N}\cup\{0\}.$ Then for $t\geq1,$%
\[
\left\Vert \nabla_{p}^{\ell}\nabla_{x}^{k}f(t)\right\Vert _{L^{2}(m_{\theta}%
)}\leq C\left\Vert f_{in}\right\Vert _{L^{2}\left(  m_{k+\ell}\right)  \ }\ ,
\]
the constant $C$ depending only upon $k$ and $\ell.$
\end{proposition}

\begin{proof}
Let $n\in\mathbb{N}$. Define
\[
H^{n}[f]\left(  t\right)  =\sum_{j=0}^{n}\left\Vert \nabla_{p}^{j}%
f(t)\right\Vert _{L^{2}(m_{\theta})}^{2}\,,\quad H_{x}^{n}[f]\left(  t\right)
=\sum_{j=0}^{n}\left\Vert \nabla_{p}^{j}\nabla_{x}f(t)\right\Vert
_{L^{2}\left(  m_{\theta}\right)  }^{2}\,.
\]
We first claim that
\[
\frac{d}{dt}H^{n}[f]\left(  t\right)  \leq-cH^{n}[f]\left(  t\right)
+C\left(  H_{x}^{n-1}[f]\left(  t\right)  +\left\Vert f(t)\right\Vert _{L^{2}%
}^{2}\right)  ,
\]
for some constants $c>0$ small and $C>0$ large. To confirm this, in view of
Lemma \ref{energy-Guo}, we find
\[
\frac{d}{dt}\left\Vert \nabla_{p}^{\ell}f\right\Vert _{L^{2}(m_{\theta})}%
^{2}\lesssim-\left\Vert \nabla_{p}^{\ell}f\right\Vert _{L_{\sigma}%
^{2}(m_{\theta})}^{2}+C\left[  \sum_{j=0}^{\ell-1}\left\Vert \nabla_{p}%
^{j}f\right\Vert _{L_{\sigma}^{2}(m_{\theta})}^{2}+\sum_{j=0}^{\ell
-1}\left\Vert \nabla_{p}^{j}\nabla_{x}f\right\Vert _{L^{2}(m_{\theta})}%
^{2}+\left\Vert f\right\Vert _{L^{2}}^{2}\right]  ,
\]
so that
\begin{align*}
\frac{d}{dt}\sum_{j=0}^{n}\varepsilon^{j}\left\Vert \nabla_{p}^{j}f\right\Vert
_{L^{2}(m_{\theta})}^{2}\,  &  \lesssim-\varepsilon^{n}\left\Vert \nabla
_{p}^{n}f\right\Vert _{L_{\sigma}^{2}(m_{\theta})}^{2}-\sum_{j=0}^{n-1}\left[
\varepsilon^{j}-C\left(  \varepsilon^{j+1}+\cdots+\varepsilon^{n}\right)
\right]  \left\Vert \nabla_{p}^{j}f\right\Vert _{L_{\sigma}^{2}(m_{\theta}%
)}^{2}\\
&  +nC\left(  \sum_{j=0}^{n-1}\left\Vert \nabla_{p}^{j}\nabla_{x}f\right\Vert
_{L^{2}\left(  m_{\theta}\right)  }^{2}+\left\Vert f\right\Vert _{L^{2}}%
^{2}\right) \\
&  \lesssim-\varepsilon^{n}\left\Vert \nabla_{p}^{n}f\right\Vert
_{L^{2}(m_{\theta})}^{2}-\sum_{j=0}^{n-1}\varepsilon^{j}\left[
1-nC\varepsilon\right]  \left\Vert \nabla_{p}^{j}f\right\Vert _{L^{2}%
(m_{\theta})}^{2}\\
&  +nC\left(  H_{x}^{n-1}[f]\left(  t\right)  +\left\Vert f\left(  t\right)
\right\Vert _{L^{2}}^{2}\right)  ,
\end{align*}
whenever we choose $0<\varepsilon\ll1$ with $1-nC\varepsilon\geq1/2$. Hence,
\begin{align*}
\frac{d}{dt}H^{n}[f]\left(  t\right)   &  \lesssim-\left\Vert \nabla_{p}%
^{n}f\right\Vert _{L^{2}(m_{\theta})}^{2}-\sum_{j=0}^{n-1}\frac{\varepsilon
^{j-n}}{2}\left\Vert \nabla_{p}^{j}f\right\Vert _{L^{2}(m_{\theta})}^{2}%
+\frac{nC}{\varepsilon^{n}}\left(  H_{x}^{n-1}[f]\left(  t\right)  +\left\Vert
f\left(  t\right)  \right\Vert _{L^{2}}^{2}\right) \\
&  \lesssim-H^{n}[f]\left(  t\right)  +C^{\prime}\left(  H_{x}^{n-1}[f]\left(
t\right)  +\left\Vert f(t)\right\Vert _{L^{2}}^{2}\right)  ,
\end{align*}
as required. Now, for $t\geq1$, we have
\begin{align*}
H^{n}[f]\left(  t\right)   &  \leq e^{-c\left(  t-1\right)  }H^{n}[f]\left(
1\right)  +C\int_{1}^{t}e^{-c\left(  t-s\right)  }\left(  H_{x}^{n-1}%
[f]\left(  s\right)  +\left\Vert f(s)\right\Vert _{L^{2}}^{2}\right)  ds\\
&  \leq e^{-c\left(  t-1\right)  }\sum_{j=0}^{n}\left\Vert \nabla_{p}%
^{j}f(1)\right\Vert _{L^{2}\left(  m_{\theta}\right)  }^{2}+C^{\prime}\left(
1-e^{-c\left(  t-1\right)  }\right)  \left\Vert f(0)\right\Vert _{L^{2}}%
^{2}+C\int_{1}^{t}e^{-c\left(  t-s\right)  }H_{x}^{n-1}[f]\left(  s\right)
ds\\
&  \lesssim\left\Vert f_{in}\right\Vert _{L^{2}\left(  m_{n}\right)  \ }%
^{2}+\int_{1}^{t}e^{-c\left(  t-s\right)  }H_{x}^{n-1}[f]\left(  s\right)  ds,
\end{align*}
the last inequality being true due to Proposition \ref{Mixed-decay}.

Next, we shall claim that
\[
H_{x}^{n-1}[f]\left(  t\right)  \leq C\left\Vert f_{in}\right\Vert
_{L^{2}\left(  m_{n}\right)  \ }^{2},\ \ \ t\geq1,
\]
for some constant $C>0$ depending only upon $n.$ When $n=1,$ it is from Lemma
\ref{x-regularityST} $\left(  ii\right)  $. When $n\geq2,$ consider the
functional
\[
\mathcal{A}_{N}\left(  t,f\right)  \equiv%
{\displaystyle\int}
\left(  \sum_{j=1}^{N+1}\left\vert \nabla_{x}^{j}f(t)\right\vert ^{2}%
+\sum_{q=1}^{N}\sum_{j=1}^{N+1-q}\varepsilon^{j+N-q}\left\vert \nabla_{p}%
^{j}\nabla_{x}^{q}f(t)\right\vert ^{2}\right)  m_{\theta}dxdp,
\]
for $\varepsilon>0$ sufficiently small. By Lemma \ref{energy-Guo} again, we
have%
\[
\frac{d}{dt}\left\Vert \nabla_{x}^{k}f\right\Vert _{L^{2}\left(  m_{\theta
}\right)  }^{2}\lesssim-\left\Vert \nabla_{x}^{k}f\right\Vert _{L_{\sigma}%
^{2}\left(  m_{\theta}\right)  }^{2}+C\left\Vert \nabla_{x}^{k}f\right\Vert
_{L^{2}}^{2},
\]
and%
\[
\frac{d}{dt}\left\Vert \nabla_{p}^{\ell}\nabla_{x}^{k}f\right\Vert
_{L^{2}(m_{\theta})}^{2}\lesssim-\left\Vert \nabla_{p}^{\ell}\nabla_{x}%
^{k}f\right\Vert _{L_{\sigma}^{2}(m_{\theta})}^{2}+C\left[  \sum_{j=0}%
^{\ell-1}\left\Vert \nabla_{p}^{j}\nabla_{x}^{k}f\right\Vert _{L_{\sigma}%
^{2}(m_{\theta})}^{2}+\sum_{j=0}^{\ell-1}\left\Vert \nabla_{p}^{j}\nabla
_{x}^{k+1}f\right\Vert _{L^{2}(m_{\theta})}^{2}+\left\Vert \nabla_{x}%
^{k}f\right\Vert _{L^{2}}^{2}\right]  .
\]
Hence, following the same argument as $H^{n}[f],$ we deduce
\[
\frac{d}{dt}\mathcal{A}_{N}\left(  t,f\right)  \leq-c\mathcal{A}_{N}\left(
t,f\right)  +C\sum_{j=1}^{N+1}\left\Vert \nabla_{x}^{j}f(t)\right\Vert
_{L^{2}}^{2},
\]
for some constants $c>0$ small and $C>0$ large as well. Therefore, for
$t\geq1$
\begin{align*}
\mathcal{A}_{N}\left(  t,f\right)   &  \leq e^{-c\left(  t-1\right)
}\mathcal{A}_{N}\left(  1,f\right)  +C\int_{1}^{t}e^{-c\left(  t-s\right)
}\sum_{j=1}^{N+1}\left\Vert \nabla_{x}^{j}f(s)\right\Vert _{L^{2}}^{2}ds\\
&  \leq e^{-c\left(  t-1\right)  }\mathcal{A}_{N}\left(  1,f\right)
+C^{\prime}\int_{1}^{t}e^{-c\left(  t-s\right)  }\sum_{j=1}^{N+1}\left\Vert
\nabla_{x}^{j}f(1)\right\Vert _{L^{2}}^{2}ds\\
&  \lesssim\mathcal{A}_{N}\left(  1,f\right)  +\sum_{j=1}^{N+1}\left\Vert
\nabla_{x}^{j}f(1)\right\Vert _{L^{2}}^{2}\lesssim\left\Vert f_{in}\right\Vert
_{L^{2}\left(  m_{N+1}\right)  \ }^{2},
\end{align*}
due to Proposition \ref{Mixed-decay}. It implies that for $t\geq1,$
\[
H_{x}^{n-1}[f]\left(  t\right)  \lesssim\mathcal{A}_{n-1}\left(  t,f\right)
\lesssim\left\Vert f_{in}\right\Vert _{L^{2}\left(  m_{n}\right)  \ }^{2}.
\]
As a consequence, we have
\[
H^{n}[f]\left(  t\right)  \lesssim\left\Vert f_{in}\right\Vert _{L^{2}\left(
m_{n}\right)  \ }^{2},\ \ \ \ H_{x}^{n}[f]\left(  t\right)  \lesssim\left\Vert
f_{in}\right\Vert _{L^{2}\left(  m_{n+1}\right)  \ }^{2}.
\]
Therefore, for $\ell\in\mathbb{N},$%
\[
\left\Vert \nabla_{p}^{\ell}f(t)\right\Vert _{L^{2}(m_{\theta})}%
\lesssim\left\Vert f_{in}\right\Vert _{L^{2}\left(  m_{\ell}\right)
\ },\ \ \ \ t\geq1.
\]
In effect, it holds that for any $k,\ \ell\in\mathbb{N}$
\[
\left\Vert \nabla_{p}^{\ell}\nabla_{x}^{k}f(t)\right\Vert _{L^{2}(m_{\theta}%
)}^{2}\lesssim\mathcal{A}_{k+\ell-1}\left(  t,f\right)  \lesssim\left\Vert
f_{in}\right\Vert _{L^{2}\left(  m_{k+\ell}\right)  \ }^{2},\ \ t\geq1.
\]
Together with Lemma \ref{x-regularityST} $\left(  ii\right)  ,$ the proof is completed.
\end{proof}

\subsection{Proof of Theorem \ref{theorem-g}}

Regarding the regularity of $g_{L,0},$ it has been done in Lemma \ref{SpecAB}.
Therefore, it remains to demonstrate that for any $k,$ $\ell\in\mathbb{N\cup
\{}0\}$,
\[
\left\Vert \nabla_{p}^{\ell}\nabla_{x}^{k}g_{L,\bot}\right\Vert _{L_{x}%
^{\infty}L_{p}^{2}}\lesssim e^{-2^{-(\ell+1)}Ct}\left(  \left\Vert
g_{in}\right\Vert _{L^{2}\left(  w_{\ell+1}\right)  }+\left\Vert
g_{in}\right\Vert _{L_{x}^{1}L_{p}^{2}}\right)
\]
and%
\[
\left\Vert \nabla_{p}^{\ell}\nabla_{x}^{k}g_{S}\right\Vert _{L_{x}^{\infty
}L_{p}^{2}}\lesssim e^{-C_{k,\ell}t}\left\Vert g_{in}\right\Vert
_{L^{2}\left(  w_{k+\ell+2}\right)  },
\]
whenever $t\geq1.$ Here $C_{k,\ell}>0$ depends only upon $k$ and $\ell.$ Among
them, the case in which $k\in\mathbb{N\cup\{}0\}$ and $\ell=0$ is a
consequence of Proposition \ref{x-regularity} and thus we may assume that
$k\in\mathbb{N\cup\{}0\}$ and $\ell\in\mathbb{N}$ in the following discussion.

Firstly, we prove that for $t\geq1,$ the $L^{2}$-norms of any derivatives of
$g_{L}$ and $g_{S}$ are bounded above uniformly in $t$ by multiples of certain
weighted $L^{2}$-norms of the initial data $g_{in}$. To see this, note that
the long wave part $g_{L}$ satisfies%
\begin{equation}
\left\{
\begin{array}
[c]{l}%
\displaystyle\partial_{t}g_{L}+\frac{1}{m_{A}}p\cdot\nabla_{x}g_{L}%
=L_{AB}g_{L}\,,\quad(t,x,p)\in({\mathbb{R}}^{+},{\mathbb{R}^{3}},{\mathbb{R}%
}^{3})\,,\\[4mm]%
\displaystyle g_{L}(0,x,p)=(g_{L})_{in}(x,p)\,,
\end{array}
\right.  \label{g.1L}%
\end{equation}
and the short wave part $g_{S}$ satisfies
\begin{equation}
\left\{
\begin{array}
[c]{l}%
\displaystyle\partial_{t}g_{S}+\frac{1}{m_{A}}p\cdot\nabla_{x}g_{S}%
=L_{AB}g_{S}\,,\quad(t,x,p)\in({\mathbb{R}}^{+},{\mathbb{R}},{\mathbb{R}}%
^{3})\,,\\[4mm]%
\displaystyle g_{S}(0,x,p)=(g_{S})_{in}(x,p),
\end{array}
\right.  \label{g.1S}%
\end{equation}
respectively, where
\[
(g_{L})_{in}(x,p)=\int_{\left\vert \eta\right\vert <\delta}e^{i\eta\cdot
x}\hat{g_{in}}\left(  \eta,p\right)  d\eta,\quad\quad\text{ }(g_{S}%
)_{in}(x,p)=\int_{\left\vert \eta\right\vert \geq\delta}e^{i\eta\cdot x}%
\hat{g_{in}}\left(  \eta,p\right)  d\eta.
\]
Hence, from Proposition \ref{energy*} (with $m_{\theta}=1$), it readily
follows that%
\[
\left\Vert \nabla_{p}^{\ell}\nabla_{x}^{k}g_{L}(t)\right\Vert _{L^{2}},\text{
}\left\Vert \nabla_{p}^{\ell}\nabla_{x}^{k}g_{S}(t)\right\Vert _{L^{2}%
}\lesssim\left\Vert g_{in}\right\Vert _{L^{2}\left(  w_{k+\ell}\right)  \ }.
\]
However,
\[
\left\Vert \nabla_{x}^{k}g_{L}(t)\right\Vert _{L^{2}}\lesssim\left\Vert
g_{in}\right\Vert _{L^{2}}%
\]
for all $k\in\mathbb{N\cup\{}0\}$. During the process of the proof of
Propositions \ref{Mixed-decay} and \ref{energy*}, we can improve the
regularity of $g_{L}$ to obtain%
\[
\left\Vert \nabla_{p}^{\ell}\nabla_{x}^{k}g_{L}(t)\right\Vert _{L^{2}}%
\lesssim\left\Vert g_{in}\right\Vert _{L^{2}\left(  w_{\ell}\right)
\ },\text{ }t\geq1.
\]
In sum, we summarize the $L^{2}$-norm bounds for derivatives as the following:

\begin{proposition}
\label{aa.7} Let $k,\ell\in\mathbb{N}\cup\{0\}$. Then for $t\geq1$
\[
\left\Vert \nabla_{p}^{\ell}\nabla_{x}^{k}g_{L}(t)\right\Vert _{L^{2}}%
\lesssim\left\Vert g_{in}\right\Vert _{L^{2}\left(  w_{\ell}\right)  \ },
\]
and%
\[
\left\Vert \nabla_{p}^{\ell}\nabla_{x}^{k}g_{S}(t)\right\Vert _{L^{2}}%
\lesssim\left\Vert g_{in}\right\Vert _{L^{2}\left(  w_{k+\ell}\right)  \ }.
\]

\end{proposition}

Thereupon, write $g_{L,\bot}=g_{L}-g_{L,0}$ and then we find
\begin{align}
\left\Vert \nabla_{p}^{\ell}\nabla_{x}^{k}g_{L,\bot}\right\Vert _{L^{2}}  &
\leq\left\Vert \nabla_{p}^{\ell}\nabla_{x}^{k}g_{L}\right\Vert _{L^{2}%
}+\left\Vert \nabla_{p}^{\ell}\nabla_{x}^{k}g_{L,0}\right\Vert _{L^{2}%
}\nonumber\\
&  \lesssim\left\Vert g_{in}\right\Vert _{L^{2}\left(  w_{\ell}\right)
\ }+(1+t)^{-(3+2k)/4}\left\Vert g_{in}\right\Vert _{L_{x}^{1}L_{p}^{2}%
}\nonumber\\
&  \lesssim\left\Vert g_{in}\right\Vert _{L^{2}\left(  w_{\ell}\right)
\ }+\left\Vert g_{in}\right\Vert _{L_{x}^{1}L_{p}^{2}} \label{g_kl-norm}%
\end{align}
for $t\geq1,$ $k,$ $\ell\in\mathbb{N\cup\{}0\}.$

Further, to obtain the decay rate, we employ the Sobolev inequality and the
interpolation inequality to attain this end. Specifically, by the Sobolev
inequality \cite[Proposition 3.8]{[Taylor]}, we find
\[
\left\Vert \nabla_{p}^{\ell}\nabla_{x}^{k}g_{L,\bot}\right\Vert _{L_{x}%
^{\infty}L_{p}^{2}}\lesssim\sqrt{\left\Vert \nabla_{p}^{\ell}\nabla_{x}%
^{k+2}g_{L,\bot}\right\Vert _{L^{2}}\left\Vert \nabla_{p}^{\ell}\nabla
_{x}^{k+1}g_{L,\bot}\right\Vert _{L^{2}}}.
\]
On the other hand, the interpolation inequality says that
\[
\left\Vert \nabla_{p}^{\ell}\nabla_{x}^{k}g_{L,\bot}\right\Vert _{L^{2}%
}\lesssim\left\{
\begin{array}
[c]{ll}%
\left\Vert \nabla_{x}^{k}g_{L,\bot}\right\Vert _{_{L^{2}}}^{1/2}\left\Vert
\nabla_{p}^{2}\nabla_{x}^{k}g_{L,\bot}\right\Vert _{_{L^{2}}}^{1/2},%
\vspace{3mm}%
& \ell=1\\
\left\Vert \nabla_{x}^{k}g_{L,\bot}\right\Vert _{L^{2}}^{2^{-\ell}}\left(
{\displaystyle\prod\limits_{j=2}^{\ell}}
\left\Vert \nabla_{p}^{j}\nabla_{x}^{k}g_{L,\bot}\right\Vert _{L^{2}%
}^{2^{-\left(  \ell-j+2\right)  }}\right)  \left\Vert \nabla_{p}^{\ell
+1}\nabla_{x}^{k}g_{L,\bot}\right\Vert _{L^{2}}^{1/2}, & \ell\geq2
\end{array}
\right.
\]
and
\begin{align*}
\left\Vert \nabla_{x}^{k}g_{L,\bot}\right\Vert _{L^{2}}  &  \lesssim\left\{
\begin{array}
[c]{ll}%
\left\Vert g_{L,\bot}\right\Vert _{_{L^{2}}}^{1/2}\left\Vert \nabla_{x}%
^{2}g_{L,\bot}\right\Vert _{_{L^{2}}}^{1/2},%
\vspace{3mm}%
& k=1\\
\left\Vert g_{L,\bot}\right\Vert _{L^{2}}^{2^{-k}}\left(
{\displaystyle\prod\limits_{j=2}^{k}}
\left\Vert \nabla_{x}^{j}g_{L,\bot}\right\Vert _{L^{2}}^{2^{-\left(
k-j+2\right)  }}\right)  \left\Vert \nabla_{x}^{k+1}g_{L,\bot}\right\Vert
_{L^{2}}^{1/2}, & k\geq2
\end{array}
\right. \\
&  \lesssim e^{-Ct}\left\Vert g_{in}\right\Vert _{L^{2}}%
\end{align*}
for all $t\geq1.$ Accordingly, combining these with $\left(  \ref{g_kl-norm}%
\right)  $ yields
\[
\left\Vert \nabla_{p}^{\ell}\nabla_{x}^{k}g_{L,\bot}\right\Vert _{L_{x}%
^{\infty}L_{p}^{2}}\lesssim e^{-2^{-(\ell+1)}Ct}\left(  \left\Vert
g_{in}\right\Vert _{L^{2}\left(  w_{\ell+1}\right)  }+\left\Vert
g_{in}\right\Vert _{L_{x}^{1}L_{p}^{2}}\right)  ,
\]
for all $t\geq1.$ Meanwhile, with the same argument, we also obtain
\begin{align*}
\left\Vert \nabla_{p}^{\ell}\nabla_{x}^{k}g_{S}\right\Vert _{L_{x}^{\infty
}L_{p}^{2}}  &  \lesssim\sqrt{\left\Vert \nabla_{p}^{\ell}\nabla_{x}%
^{k+1}g_{S}\right\Vert _{L^{2}}\left\Vert \nabla_{p}^{\ell}\nabla_{x}%
^{k+2}g_{S}\right\Vert _{L^{2}}}\\
&  \lesssim e^{-2^{-\left(  k+\ell+2\right)  }Ct}\left\Vert g_{in}\right\Vert
_{L^{2}\left(  w_{k+\ell+2}\right)  \ },
\end{align*}
by $\left(  \ref{aa.8}\right)  $ and Proposition \ref{aa.7}. This completes
the proof of Theorem \ref{theorem-g}.

\section{The result for $h$}

Recall the equation
\begin{equation}
\left\{
\begin{array}
[c]{l}%
\displaystyle\pa_{t}h+\frac{1}{m_{B}}p\cdot\nabla_{x}h=L_{BB}h+L_{BA}%
g\,,\quad(t,x,p)\in({\mathbb{R}}^{+},{\mathbb{R}}^{3},{\mathbb{R}}^{3})\,,\\[4mm]%
\displaystyle h(0,x,p)=h_{in}(x,p)\,.
\end{array}
\right.  \label{h.1}%
\end{equation}
In previous sections, the regularity of the solution $g$ has been investigated
well. Now, regarding $g$ as a source term, we shall devote to establish the
regularization estimate for the solution $h$ in the rest of the article.
Especially, we demonstrate that $h$ has the same decay rate as $g$ (see
Theorem \ref{MAIN}).

Let $\mathbb{G_{BB}}^{t}$ be the solution operator of
\[
\displaystyle\pa_{t}f+\frac{1}{m_{B}}p\cdot\nabla_{x}f=L_{BB}f\,,
\]
and then the homogeneous part of $h$ is $\widetilde{h}=\mathbb{G_{BB}}%
^{t}h_{in}$. Concerning estimates on the homogeneous part, it is the same as
those on $g_{L}$ and $g_{S}$ in previous sections and hence we only state the
results without proofs. As before, we decompose $\widetilde{h}$ into the long
wave-fluid part $\widetilde{h}_{L,0}$, long wave-nonfluid part $\widetilde
{h}_{L,\bot}$ and short wave part $\widetilde{h}_{S}$. Then they respectively
satisfy
\[
\left\Vert \nabla_{p}^{\ell}\nabla_{x}^{k}\widetilde{h}_{L,0}\right\Vert
_{L_{x}^{\infty}L_{p}^{2}}\lesssim(1+t)^{-(3+k)/2}\left\Vert h_{in}\right\Vert
_{L_{x}^{1}L_{p}^{2}}\,,
\]

\[
\left\Vert \nabla_{p}^{\ell}\nabla_{x}^{k}\widetilde{h}_{L,\bot}\right\Vert
_{L_{x}^{\infty}L_{p}^{2}}\lesssim e^{-Ct}\left(  \left\Vert h_{in}\right\Vert
_{L^{2}\left(  w_{\ell+1}\right)  }+\left\Vert h_{in}\right\Vert _{L_{x}%
^{1}L_{p}^{2}}\right)  ,
\]
and
\[
\left\Vert \nabla_{p}^{\ell}\nabla_{x}^{k}\widetilde{h}_{S}\right\Vert
_{L_{x}^{\infty}L_{p}^{2}}\lesssim e^{-Ct}\left\Vert h_{in}\right\Vert
_{L^{2}\left(  w_{k+\ell+2}\right)  },
\]
with the constant $C>0$ depending on $k$ and $\ell$, and
\[
w_{n}\equiv%
\begin{cases}
1, & \ga\in\left[  0,1\right]  ,\\
\left\langle p\right\rangle ^{|\ga|n}, & \ga\in\left[  -2,0\right)  .
\end{cases}
\]
Thereupon, it remains to control the inhomogeneous part. Hereafter, we may
assume $h_{in}=0$.

We now decompose the solution $h$ into the long wave part $h_{L}$ and the
short wave part $h_{S}$, which respectively satisfy
\begin{equation}
\left\{
\begin{array}
[c]{l}%
\displaystyle\pa_{t}h_{L}+\frac{1}{m_{B}}p\cdot\nabla_{x}h_{L}=L_{BB}%
h_{L}+L_{BA}g_{L}\,,\quad(t,x,p)\in({\mathbb{R}}^{+},{\mathbb{R}}^{3},{\mathbb{R}%
}^{3})\,,\\[4mm]%
\displaystyle h_{L}(0,x,p)=0\,,
\end{array}
\right.  \label{h.1L}%
\end{equation}
and
\begin{equation}
\left\{
\begin{array}
[c]{l}%
\displaystyle\pa_{t}h_{S}+\frac{1}{m_{B}}p\cdot\nabla_{x}h_{S}=L_{BB}%
h_{S}+L_{BA}g_{S}\,,\quad(t,x,p)\in({\mathbb{R}}^{+},{\mathbb{R}}^{3},{\mathbb{R}%
}^{3})\,,\\[4mm]%
\displaystyle h_{S}(0,x,p)=0\,.
\end{array}
\right.  \label{h.1S}%
\end{equation}
It is easy to see that
\begin{align*}
h_{L}(t,x,p)  &  =\int_{0}^{t}\int_{|\eta|<\de}e^{i\eta\cdot x}e^{(-ip\cdot
\eta/m_{B}+L_{BB})(t-s)}L_{BA}e^{(-ip\cdot\eta/m_{A}+L_{AB})s}\hat{g}%
_{in}d\eta ds\\
&  =h_{00}^{L}+h_{0\perp}^{L}+h_{\perp0}^{L}+h_{\perp\perp}^{L}\,,
\end{align*}
where
\begin{align*}
h_{00}^{L}  &  =\sum_{j=0}^{4}\int_{0}^{t}\int_{|\eta|<\de}e^{i\eta\cdot
x}e^{\sigma_{j}(\eta)(t-s)+\la(\eta)s}\big<e_{j}(-\eta),L_{BA}e_{D}%
(\eta)\big>_{p}\big<\hat{g}_{in},e_{D}(-\eta)\big>_{p}e_{j}(\eta)d\eta ds\\
h_{0\perp}^{L}  &  =\sum_{j=0}^{4}\int_{0}^{t}\int_{|\eta|<\de}e^{i\eta\cdot
x}e^{\sigma_{j}(\eta)(t-s)}\big<e_{j}(-\eta),L_{BA}e^{(-ip\cdot\eta
/m_{A}+L_{AB})s}\Pi_{\eta}^{D\perp}\hat{g}_{in}\big>_{p}e_{j}(\eta)d\eta ds\\
h_{\perp0}^{L}  &  =\int_{0}^{t}\int_{|\eta|<\de}e^{i\eta\cdot x}%
e^{\la(\eta)s}e^{(-ip\cdot\eta/m_{B}+L_{BB})(t-s)}\Pi_{\eta}^{\perp}%
L_{BA}e_{D}(\eta)\big<e_{D}(-\eta),\hat{g}_{in}\big>_{p}d\eta ds\\
h_{\perp\perp}^{L}  &  =\int_{0}^{t}\int_{|\eta|<\de}e^{i\eta\cdot
x}e^{(-ip\cdot\eta/m_{B}+L_{BB})(t-s)}\Pi_{\eta}^{\perp}L_{BA}e^{(-ip\cdot
\eta/m_{A}+L_{AB})s}\Pi_{\eta}^{D\perp}\hat{g}_{in}d\eta ds\,.
\end{align*}
Here the eigenfunctions $\{e_{j}\left(  \eta\right)  \}_{j=0}^{4}\ $ and
$e_{D}\left(  \eta\right)  $ are defined as in $\left(  \ref{pre.bb.h}%
\ \right)  $ and $\left(  \ref{pre.ab.h}\right)  $ with $\left\langle
e_{j}(-\eta),e_{l}(\eta)\right\rangle _{p}=$ $\delta_{jl}$ and $\left\langle
e_{D}(-\eta),e_{D}(\eta)\right\rangle _{p}=1.$ The main result for the
solution $h$ is stated as below.

\begin{theorem}
\label{theorem-h} \textrm{(Main result of $h$)} Let $k\ $and$\ \ell\ $be
nonnegative integers. Then for $t\geq1$
\[
\Vert\nabla_{p}^{\ell}\nabla_{x}^{k}(h_{00}^{L}+h_{0\perp}^{L})\Vert
_{L_{x}^{\infty}L_{p}^{2}}\leq(1+t)^{-(3+k)/2}\Vert g_{in}\Vert_{L_{x}%
^{1}L_{p}^{2}}\,,
\]

\[
\left\Vert \nabla_{p}^{\ell}\nabla_{x}^{k}h_{\perp0}^{L}\right\Vert
_{L_{x}^{\infty}L_{p}^{2}}\lesssim(1+t)^{-(3+k)/2}\left(  \left\Vert
g_{in}\right\Vert _{L^{2}\ }+\Vert g_{in}\Vert_{L_{x}^{1}L_{p}^{2}}\right)  ,
\]
and
\[
\left\Vert \nabla_{p}^{\ell}\nabla_{x}^{k}h_{\perp\perp}^{L}\right\Vert
_{L_{x}^{\infty}L_{p}^{2}}\lesssim e^{-Ct}\Vert g_{in}\Vert_{L^{2}\left(
w_{\ell+1}\right)  }.
\]
On the other hand,
\[
\left\Vert \nabla_{p}^{\ell}\nabla_{x}^{k}h_{S}\right\Vert _{L_{x}^{\infty
}L_{p}^{2}}\lesssim e^{-Ct}\left\Vert g_{in}\right\Vert _{L^{2}\left(
w_{k+\ell+2}\right)  },
\]
the constant $C>0$ depending on $k$ and $\ell$. Here
\[
w_{n}\equiv%
\begin{cases}
1, & \ga\in\left[  0,1\right]  ,\\
\left\langle p\right\rangle ^{|\ga|n}, & \ga\in\left[  -2,0\right)  .
\end{cases}
\]

\end{theorem}

\subsection{Improvement of the $x$-regularity}

In this subsection, we deal with the $x$-regularization estimate for the long
wave part $h_{L}$. Owing to the fact that the system $\left(  \ref{mix.2.a}%
\right)  $ is decoupled, we instinctively use the Duhamel principle to solve
$h$ through treating $g$ as a source term. However, there are some possible
wave resonances between $g$ and $h$ as mentioned in the Introduction, and thus
this method usually leads the solution $h$ to lower decay rate than $g.$ To
remedy this, further physical properties of the collision operators, namely
the microscopic cancellations representing the conservation of mass, total
momentum and total energy, are employed in the following discussion.

\begin{proposition}
\label{long-wave-h}\textrm{(Long wave $h_{L}$, improve of the $x$ regularity)
}For $k,\ell\in\mathbb{N\cup}\{0\}$
\[
\Vert\nabla_{p}^{\ell}\nabla_{x}^{k}(h_{00}^{L}+h_{0\perp}^{L})\Vert
_{L_{x}^{\infty}L_{p}^{2}}\lesssim(1+t)^{-(3+k)/2}\Vert g_{in}\Vert_{L_{x}%
^{1}L_{p}^{2}}\,,
\]

\begin{equation}
\Vert\nabla_{x}^{k}h^{L}_{\perp0}\Vert_{L_{x}^{\infty}L_{p}^{2}}%
\lesssim(1+t)^{-(4+k)/2}\Vert g_{in}\Vert_{L_{x}^{1}L_{p}^{2}}\,,
\end{equation}
and
\begin{equation}
\Vert\nabla_{x}^{k}h^{L}_{\perp\perp}\Vert_{L_{x}^{\infty}L_{p}^{2}}\lesssim
te^{-a(\tau)t}\Vert g_{in}\Vert_{L_{x}^{1}L_{p}^{2}}\,.
\end{equation}

\end{proposition}

\begin{proof}
We first split $h_{00}^{L}=I_{0}+I_{1}+I_{2}+I_{3}+I_{4}$, where each $I_{j}$
is the orthogonal projection of $h_{00}^{L}$ along $e_{j},$ i.e.,%

\[
I_{j}=\int_{0}^{t}\int_{|\eta|<\delta}e^{i\eta\cdot x}e^{\sigma_{j}%
(\eta)(t-s)+\lambda(\eta)s}\left\langle e_{j}(-\eta),L_{BA}e_{D}%
(\eta)\right\rangle _{p}\left\langle \hat{g}_{in},e_{D}(-\eta)\right\rangle
_{p}e_{j}(\eta)d\eta ds\,.
\]
Notice that $\lambda\left(  \eta\right)  -\sigma_{j}(\eta)\neq0$ for
$j=0\ $and $1.$ It is natural to evaluate $I_{j}$ ($j=0,1$) by the fundamental
theorem of calculus, that is,
\begin{equation}
\int_{0}^{t}e^{\sigma_{j}(\eta)(t-s)+\lambda(\eta)s}ds=\frac{1}{\lambda\left(
\eta\right)  -\sigma_{j}(\eta)}\left[  e^{\lambda\left(  \eta\right)
t}-e^{\sigma_{j}(\eta)t}\right]  ,\ \ \ j=0,1.\bigskip\label{1st-FTOC}%
\end{equation}
On the other hand, we have such a \textquotedblleft microscopic
cancellation\textquotedblright\ $L_{BA}E_{D}=0$ from $Q_{BA}\left(
M_{B},M_{A}\right)  =0\ $that
\[
\left\langle e_{j}(-\eta),L_{BA}e_{D}(\eta)\right\rangle _{p}=\frac
{1}{\left\Vert E_{D}+iE_{D,1}\eta+O(|\eta|^{2})\right\Vert _{L_{p}^{2}}%
}\left\langle e_{j}(-\eta),iL_{BA}E_{D,1}\eta+O(|\eta|^{2})\right\rangle
_{p}=O\left(  \left\vert \eta\right\vert \right)  .
\]

Therefore,
\[
\Vert\nabla_{x}^{k}I_{j}\Vert_{L_{x}^{\infty}L_{p}^{2}}\leq C\int_{|\eta
|<\de}\frac{|\eta|^{k+1}}{|\la(\eta)-\sigma_{j}(\eta)|}\Big|e^{\la(\eta
)t}+e^{\sigma_{j}(\eta)t}\Big||\hat{g}_{in}|_{L_{p}^{2}}d\eta\,,
\]
for all $k\in\mathbb{N\cup\{}0\}.$ Further, since
\begin{equation}
\frac{|\eta|}{|\lambda(\eta)-\sigma_{j}(\eta)|}=\frac{1}{|(-a_{2}+a_{j,2}%
)\eta+ia_{j,1}+O(|\eta|^{2})|}=O(1),\ \ j=0,1,\label{eigen-difference}%
\end{equation}
we deduce
\begin{align*}
\Vert\nabla_{x}^{k}I_{0}\Vert_{L_{x}^{\infty}L_{p}^{2}},\ \Vert\nabla_{x}%
^{k}I_{1}\Vert_{L_{x}^{\infty}L_{p}^{2}} &  \leq C\int_{|\eta|<\de}|\eta
|^{k}\Big|e^{\la(\eta)t}+e^{\sigma_{j}(\eta)t}\Big|d\eta\Vert g_{in}%
\Vert_{L_{x}^{1}L_{p}^{2}}\\
&  \lesssim(1+t)^{-\left(  k+3\right)  /2}\Vert g_{in}\Vert_{L_{x}^{1}%
L_{p}^{2}}\,.
\end{align*}

In contrast with $I_{0}$ and $I_{1},$ we lose $\left(  \ref{1st-FTOC}\right)
$ or $\left(  \ref{eigen-difference}\right)  $ in the case of $I_{j},\ j=1,2$
and $3.\ $To maintain the decay rate of $h,$ we come up with further
microscopic cancellations as well as $L_{BA}E_{D}=0.$ In the light of the
conservation of mass and total energy, the operator
\[
L_{AB}J+\frac{m_{A}}{m_{B}}\frac{\sqrt{M_{B}}}{\sqrt{M_{A}}}L_{BA}J
\]
is orthogonal to the collision invariants $\sqrt{M_{A}}$ and $|p|^{2}%
\sqrt{M_{A}}$. With this applied to $J=E_{D,1}=L_{AB}^{-1}\left(
p\cdot\om/m_{A}\right)  E_{D}$ (where we recall that $E_{D}=\sqrt{M_{A}}$) it
follows
\[
\int_{{\mathbb{R}}^{3}}\Big(L_{AB}\big[L_{AB}^{-1}\left(  p\cdot
\om/m_{A}\right)  E_{D}\big]+\frac{m_{A}}{m_{B}}\frac{\sqrt{M_{B}}}%
{\sqrt{M_{A}}}L_{BA}\big[L_{AB}^{-1}\left(  p\cdot\om/m_{A}\right)
E_{D}\big]\Big)\Phi_{1}\sqrt{M_{A}}dp=0,
\]
for any linear combination $\Phi_{1}$ of $1$\ and $|p|^{2},$ so that
\[
\frac{m_{A}}{m_{B}}\int_{{\mathbb{R}}^{3}}\Phi_{1}\sqrt{M_{B}}L_{BA}%
\big[L_{AB}^{-1}\left(  p\cdot\om/m_{A}\right)  E_{D}\big]dp=-\int
_{{\mathbb{R}}^{3}}\Phi_{1}\left(  p\cdot\om/m_{A}\right)  M_{A}dp\,.
\]
Choosing $\Phi_{1}=E_{2}/\sqrt{M_{B}}=\sqrt{\frac{1}{10}}(-5+|p|^{2})$, we
have
\[
\int_{{\mathbb{R}}^{3}}E_{2}L_{BA}\big[L_{AB}^{-1}\left(  p\cdot
\om/m_{A}\right)  E_{D}\big]dp=-\frac{m_{B}}{m_{A}}\int_{{\mathbb{R}}^{3}%
}\sqrt{\frac{1}{10}}(-5+|p|^{2})\left(  p\cdot\om/m_{A}\right)  M_{A}dp=0,
\]
since the integrand is odd. It turns out that $\big<e_{2}(-\eta),L_{BA}%
e_{D}(\eta)\big>_{p}=O(|\eta|^{2})$. Hence
\begin{align*}
\Vert\nabla_{x}^{k}I_{2}\Vert_{L_{x}^{\infty}L_{p}^{2}}  &  \leq C\int_{0}%
^{t}\int_{|\eta|<\de}\left\vert \eta\right\vert ^{2+k}\left\vert e^{\sigma
_{2}(\eta)(t-s)+\la(\eta)s}\right\vert |\hat{g}_{in}|_{L_{p}^{2}}d\eta ds\\
&  \leq C\int_{0}^{t}\int_{|\eta|<\de}\left\vert \eta\right\vert
^{2+k}e^{-\overline{a}|\eta|^{2}t}d\eta ds\Vert g_{in}\Vert_{L_{x}^{1}%
L_{p}^{2}}\ \ \ \\
&  \leq C^{\prime}\int_{0}^{t}\left(  1+t\right)  ^{-\left(  5+k\right)
/2}ds\Vert g_{in}\Vert_{L_{x}^{1}L_{p}^{2}}\\
&  \leq C^{\prime\prime}\left(  1+t\right)  ^{-\left(  3+k\right)  /2}\Vert
g_{in}\Vert_{L_{x}^{1}L_{p}^{2}},
\end{align*}
here $\overline{a}=1/2\min\{a_{2},a_{2,2}\}>0$ whenever $\delta>0$ is small.

Furthermore, owing to the conservation of momentum, the operator
\[
L_{AB}J+\frac{\sqrt{M_{B}}}{\sqrt{M_{A}}}L_{BA}J
\]
is orthogonal to all collision invariants $p_{1}\sqrt{M_{A}},\ p_{2}%
\sqrt{M_{A}}$\ and $p_{3}\sqrt{M_{A}}$. With this applied to $J=E_{D,1}%
=L_{AB}^{-1}\left(  p\cdot\om/m_{A}\right)  E_{D}$, it follows
\[
\int_{{\mathbb{R}}^{3}}\Big(L_{AB}\big[L_{AB}^{-1}\left(  p\cdot
\om/m_{A}\right)  E_{D}\big]+\frac{\sqrt{M_{B}}}{\sqrt{M_{A}}}L_{BA}%
\big[L_{AB}^{-1}\left(  p\cdot\om/m_{A}\right)  E_{D}\big]\Big)\Phi_{2}%
\sqrt{M_{A}}dp=0,
\]
for any linear combination $\Phi_{2}$ of $p_{1},\ p_{2}$ and $p_{3}$, so that
\[
\int_{{\mathbb{R}}^{3}}\Phi_{2}\sqrt{M_{B}}L_{BA}\big[L_{AB}^{-1}\left(
p\cdot\om/m_{A}\right)  E_{D}\big]dp=-\int_{{\mathbb{R}}^{3}}\Phi_{2}\left(
p\cdot\om/m_{A}\right)  M_{A}dp.
\]
Choosing $\Phi_{2}=E_{3}/\sqrt{M_{B}}=\omega_{1}^{\bot}\cdot p,$ we find
\begin{align*}
&  \int_{{\mathbb{R}}^{3}}E_{3}L_{BA}\big[L_{AB}^{-1}\left(  p\cdot
\om/m_{A}\right)  E_{D}\big]dp\\
&  =-\int_{{\mathbb{R}}^{3}}\left(  p\cdot\om_{1}^{\bot}\right)  \left(
p\cdot\om\right)  \frac{M_{A}}{m_{A}}dp=-\sum_{i,j=1}^{3}\left(  \int
p_{i}p_{j}\frac{M_{A}}{m_{A}}dp\right)  \left(  \omega_{1}^{\bot}\right)
_{i}\omega_{j}=0,
\end{align*}
due to the fact that $\int p_{i}p_{j}\frac{M_{A}}{m_{A}}dp=\frac{1}{3m_{B}%
}\delta_{ij}$ and $\omega_{1}^{\bot}\bot\omega.$ It turns out that
$\big<e_{3}(-\eta),L_{BA}e_{D}(\eta)\big>_{p}=O(|\eta|^{2})$. Likewise,
$\big<e_{4}(-\eta),L_{BA}e_{D}(\eta)\big>_{p}=O(|\eta|^{2})$. Therefore,
\[
\Vert\nabla_{x}^{k}I_{3}\Vert_{L_{x}^{\infty}L_{p}^{2}},\Vert\nabla_{x}%
^{k}I_{4}\Vert_{L_{x}^{\infty}L_{p}^{2}}\lesssim\left(  1+t\right)  ^{-\left(
3+k\right)  /2}\Vert g_{in}\Vert_{L_{x}^{1}L_{p}^{2}}\,.
\]
This completes the estimate of $h_{00}^{L}$.

For $h_{0\perp}^{L}$, direct computation gives
\begin{align*}
\Vert\nabla_{x}^{k}h_{0\perp}^{L}\Vert_{L_{x}^{\infty}L_{p}^{2}}  &  \leq
C\sum_{j=1}^{3}\int_{0}^{t}\int_{|\eta|<\de}\left\vert \eta\right\vert
^{k}e^{-a(\tau_{2})s}\left\vert e^{\sigma_{j}(\eta)(t-s)}\right\vert |\hat
{g}_{in}|_{L_{p}^{2}}d\eta ds\\
&  \leq C^{\prime}\left(  \int_{0}^{t/2}e^{-a(\tau_{2})s}(1+t)^{-\left(
3+k\right)  /2}ds+\int_{t/2}^{t}e^{-a(\tau_{2})s}ds\right)  \Vert g_{in}%
\Vert_{L_{x}^{1}L_{p}^{2}}\,\\
&  \lesssim\left(  1+t\right)  ^{-\left(  3+k\right)  /2}\Vert g_{in}%
\Vert_{L_{x}^{1}L_{p}^{2}}.
\end{align*}
In fact, one can improve $p$ regularity immediately; that is, for any
$\ell,\ k\in\mathbb{N\cup\{}0\mathbb{\}}$,
\[
\Vert\nabla_{p}^{\ell}\nabla_{x}^{k}\left(  h_{00}^{L}+h_{0\perp}^{L}\right)
\Vert_{L_{x}^{\infty}L_{p}^{2}}\lesssim(1+t)^{-\left(  3+k\right)  /2}\Vert
g_{in}\Vert_{L_{x}^{1}L_{p}^{2}}\,.
\]
For $h_{\perp0}^{L}$, we have
\begin{align*}
\Vert\nabla_{x}^{k}h_{\perp0}^{L}\Vert_{L_{x}^{\infty}L_{p}^{2}}  &
\lesssim\int_{0}^{t}\int_{|\eta|<\delta}|\eta|^{k+1}e^{-a(\tau_{1}%
)(t-s)}\left\vert e^{\lambda(\eta)s}\right\vert |\hat{g}_{in}|_{L_{p}^{2}%
}d\eta ds\\
&  \lesssim\left(  \int_{t/2}^{t}e^{-a(\tau_{1})(t-s)}\left(  1+s\right)
^{-\left(  k+4\right)  /2}ds+\int_{0}^{t/2}e^{-a(\tau_{1})\left(  t-s\right)
}ds\right)  \Vert g_{in}\Vert_{L_{x}^{1}L_{p}^{2}}\\
&  \lesssim(1+t)^{-\left(  4+k\right)  /2}\Vert g_{in}\Vert_{L_{x}^{1}%
L_{p}^{2}}\,.
\end{align*}
Finally,
\[
\Vert\nabla_{x}^{k}h_{\perp\perp}^{L}\Vert_{L_{x}^{\infty}L_{p}^{2}}%
\lesssim\int_{0}^{t}\int_{|\eta|<\delta}\left\vert \eta\right\vert
^{k}e^{-a(\tau_{1})(t-s)-a(\tau_{2})s}d\eta ds\Vert g_{in}\Vert_{L_{x}%
^{1}L_{p}^{2}}\lesssim te^{-a(\tau)t}\Vert g_{in}\Vert_{L_{x}^{1}L_{p}^{2}},
\]
here $a(\tau)=\min\{a(\tau_{1}),a(\tau_{2})\}.$ This completes the proof of
the proposition.
\end{proof}

\subsection{Improvement of the $x$ and $p$ regularities (general discussion)}

Let $f$ be a solution of the linearized Landau equation
\begin{equation}
\left\{
\begin{array}
[c]{l}%
\partial_{t}f+\frac{1}{m_{B}}p\cdot\nabla_{x}f=L_{BB}f+L_{BA}g,%
\vspace{3mm}%
\\
f\left(  0,x,p\right)  =f_{in}\left(  x,p\right)  .
\end{array}
\right.  \label{fBB}%
\end{equation}
where $g$ is known and satisfies the equation $\partial_{t}g+\frac{1}{m_{A}%
}p\cdot\nabla_{x}g=L_{AB}g.$ In the following, we first apply similar
arguments in section $3$ to establish the regularization estimates for $f$ in
short time and in large time. After that, based on these regularization
estimates, one can readily obtain the decay rate of $h_{S}$ through the
interpolation trick as before. Contrary to $h_{S}$, the $L_{x}^{\infty}L_{p}^{2}$
norm of $\nabla_{x}h^L_{\bot0}$ decays only algebraically, rather than exponentially. To maintain the decay rate of $\nabla_{p}^{\ell}\nabla_{x}^{k} h^L_{\bot0}$ the same as that of homogeneous problem, we interpolate  the $L_{x}^{\infty}L_{p}^{2}$ norm of it from higher $p$-derivatives. However, the method used before for improving $p$-regularity results in  the weights imposed on initial data grow drastically (In fact, exponentially grows with respect to $\ell$).  We refine the $p$-regularization estimate by making use of equation \eqref{eq:improve**} and  the $p$-regularity of source term $g_{L;0}$. In such a way, the $p$-regularity of $h^{L}_{\bot 0}$ has been improved without imposing any weight on initial data.

Let the operator $\mathcal{L}_{BB}=-\frac{1}{m_{B}}p\cdot\nabla_{x}%
-\Lambda^{BB}.$ As shown in \cite{[10]}, the operator $e^{t\mathcal{L}_{BB}}$
has the same energy estimate and regularization estimate as
$e^{t\mathcal{L}_{AB}}$ in Lemma \ref{regularization}. Based on this estimate,
we deduce the $x$-regularity of $f$ in small time as below.

\begin{lemma}
\label{hx-regularityST}Let $f$ be a solution of the linearized Landau equation
$\left(  \ref{fBB}\right)  .$ Then for $k\in\mathbb{N\cup\{}0\},$

\noindent\textrm{(i)}\quad there is $C_{k}>0$ such that for $0<t\leq1$%
\[
\left\Vert \nabla_{x}^{k}f\right\Vert _{L^{2}\left(  m_{\theta}\right)  }\leq
C_{k}t^{-\frac{3}{2}k}\left(  \left\Vert f_{in}\right\Vert _{L^{2}\left(
m_{k}\right)  }+\left\Vert g_{in}\right\Vert _{L^{2}\left(  m_{k}\right)
}\right)  .
\]
\noindent\textrm{(ii)}\quad For $t\geq1,$%
\[
\left\Vert \nabla_{x}^{k}f\right\Vert _{L^{2}\left(  m_{\theta}\right)
}\lesssim\left(  \left\Vert f_{in}\right\Vert _{L^{2}\left(  m_{k}\right)
}+t^{1/2}\left\Vert g_{in}\right\Vert _{L^{2}\left(  m_{k}\right)  }\right)
.
\]

\end{lemma}

\begin{proof}
Firstly, similar to $\left(  \ref{f-weighted-energy}\right)  $ and $\left(
\ref{fx-weighted-energy}\right)  $, there exists a universal constant $C>0$
such that
\begin{equation}
\left\Vert \mathbb{G_{BB}}^{t}f_{in}\right\Vert _{L^{2}\left(  m_{\theta
}\right)  }\leq C\left\Vert f_{in}\right\Vert _{L^{2}\left(  m_{\theta
}\right)  }, \label{f-energy-BB1}%
\end{equation}
and
\begin{equation}
\left\Vert \partial_{x}^{\alpha}\mathbb{G_{BB}}^{s_{2}}f_{in}\right\Vert
_{L^{2}\left(  m_{\theta}\right)  }\leq C\left\Vert \partial_{x}^{\alpha
}\mathbb{G_{BB}}^{s_{1}}f_{in}\right\Vert _{L^{2}\left(  m_{\theta}\right)  },
\label{fx-energy-BB1}%
\end{equation}
for any $0<s_{1}<s_{2}.$ By the Duhamel principle,
\[
f=\mathbb{G_{BB}}^{t}f_{in}+\int_{0}^{t}\mathbb{G_{BB}}^{t-s}L_{BA}g\left(
s\right)  ds.
\]
Notice that $L_{BA}$ is an integral operator like $K^{BB}.$ Hence, for $k=0,$
we find
\begin{align}
\left\Vert f\right\Vert _{L^{2}\left(  m_{\theta}\right)  }  &  \lesssim
\left(  \left\Vert \mathbb{G_{BB}}^{t}f_{in}\right\Vert _{L^{2}\left(
m_{\theta}\right)  }^{2}+\int_{0}^{t}\left\Vert \mathbb{G_{BB}}^{t-s}%
L_{BA}g\left(  s\right)  \right\Vert _{L^{2}\left(  m_{\theta}\right)  }%
^{2}ds\right)  ^{1/2}\label{f-energy-BB2}\\
&  \lesssim\left(  \left\Vert f_{in}\right\Vert _{L^{2}\left(  m_{\theta
}\right)  }^{2}+\int_{0}^{t}\left\Vert g\left(  s\right)  \right\Vert
_{L^{2}\left(  m_{\theta}\right)  }^{2}ds\right)  ^{1/2}\nonumber\\
&  \lesssim\left\Vert f_{in}\right\Vert _{L^{2}\left(  m_{\theta}\right)
}+t^{1/2}\left\Vert g_{in}\right\Vert _{L^{2}\left(  m_{\theta}\right)
},\nonumber
\end{align}
from $\left(  \ref{f-weighted-energy}\right)  $ and $\left(
\ref{f-energy-BB1}\right)  $. In addition, since $\partial_{x}^{\alpha}$ is
commutative with equation $\left(  \ref{fBB}\right)  ,$ we as well have
\begin{equation}
\left\Vert \partial_{x}^{\alpha}f\left(  s_{2}\right)  \right\Vert
_{L^{2}\left(  m_{\theta}\right)  }\lesssim\left\Vert \partial_{x}^{\alpha
}f\left(  s_{1}\right)  \right\Vert _{L^{2}\left(  m_{\theta}\right)
}+\left(  s_{2}-s_{1}\right)  ^{1/2}\left\Vert \partial_{x}^{\alpha}g\left(
s_{1}\right)  \right\Vert _{L^{2}\left(  m_{\theta}\right)  },
\label{fx-energy-BB2}%
\end{equation}
for any $0<s_{1}<s_{2}.$

For $k\geq1,$ using a Picard type iteration introduced in Lemma
\ref{x-regularityST}, we rewrite%
\[
f=f^{\left(  0\right)  }+f^{\left(  1\right)  }+\cdots+f^{\left(  2k\right)
}+\mathcal{R}^{\left(  k\right)  },
\]
where $f^{\left(  0\right)  }\ $satisfies
\[
\left\{
\begin{array}
[c]{l}%
\partial_{t}f^{(0)}=\mathcal{L}_{BB}f^{\left(  0\right)  },%
\vspace{3mm}%
\\
f^{(0)}(0,x,p)=f_{in}\left(  x,p\right)  ,
\end{array}
\right.
\]
$f^{\left(  j\right)  }$ , $1\leq j\leq2k,$ satisfies
\[
\left\{
\begin{array}
[c]{l}%
\partial_{t}f^{(j)}=\mathcal{L}_{BB}f^{(j)}+K^{BB}f^{(j-1)}\,,\\[4mm]%
f^{(j)}(0,x,p)=0\,,
\end{array}
\right.
\]
and $\mathcal{R}^{\left(  k\right)  }$ solves the equation%
\[
\left\{
\begin{array}
[c]{l}%
\partial_{t}\mathcal{R}^{\left(  k\right)  }=L_{BB}\mathcal{R}^{\left(
k\right)  }+K^{BB}f^{(2k)}\,+L_{BA}g\\[4mm]%
\mathcal{R}^{\left(  k\right)  }(0,x,p)=0.
\end{array}
\right.
\]
Under this decomposition, following the same procedure as in Lemma
\ref{x-regularityST}, we finish the proof.
\end{proof}

\begin{proposition}
\label{short-time-h} Let $f$ be a solution to equation $\left(  \ref{fBB}%
\right)  $ and let $k,$ $\ell\in\mathbb{N}\cup\{0\}.$ Then for $0<t\leq1$
\[
\left\Vert \nabla_{p}^{\ell}\nabla_{x}^{k}f\left(  t\right)  \right\Vert
_{L^{2}\left(  m_{\theta}\right)  }\leq Ct^{-\frac{\ell}{2}-\frac{3}{2}%
(k+\ell) }\left(  \left\Vert f_{in}\right\Vert _{L^{2}\left(  m_{k+\ell
}\right)  \ }+\left\Vert g_{in}\right\Vert _{L^{2}\left(  m_{k+\ell}\right)
\ }\right)  .
\]

\end{proposition}

\begin{proof}
Let $u$ be a solution of equation
\begin{equation}
\left\{
\begin{array}
[c]{l}%
\pa_{t}u=\mathcal{L}_{BB}u\,,\\[4mm]%
u(0,x,p)=u_{0}(x,p).
\end{array}
\right.
\end{equation}
Then for any $\ell\in\mathbb{N}$
\begin{equation}
t\Vert\nabla_{p}^{\ell}u\Vert_{L^{2}\left(  m_{\theta}\right)  }^{2}%
\lesssim\mathcal{F}_{\ell}\left(  0,u\right)  ,\text{ }0<t\leq1
\label{L-estimate}%
\end{equation}
where $\mathcal{F}_{\ell}\left(  t,u\right)  $ is defined as in the proof of
Proposition \ref{Mixed-decay}.

Let $0<t_{0}\leq1.$ By Duhamel's principle,
\[
f\left(  t\right)  =e^{(t-t_{0}/2)\mathcal{L}_{BB}}f(t_{0}/2)+\int_{t_{0}%
/2}^{t}e^{\left(  t-s\right)  \mathcal{L}_{BB}}K^{BB}f\left(  s\right)
ds+\int_{t_{0}/2}^{t}e^{\left(  t-s\right)  \mathcal{L}_{BB}}L_{AB}g\left(
s\right)  ds,
\]
for $0<t_{0}/2\leq t\leq1,$ so that
\begin{align*}
\left\Vert \nabla_{p}f\left(  t\right)  \right\Vert _{L^{2}\left(  m_{\theta
}\right)  }  &  \leq C\left(  t-t_{0}/2\right)  ^{-\frac{1}{2}}\left(
\left\Vert f\left(  t_{0}/2\right)  \right\Vert _{L^{2}\left(  m_{1}\right)
}+\left\Vert \nabla_{x}f\left(  t_{0}/2\right)  \right\Vert _{L^{2}\left(
m_{\theta}\right)  }\right) \\
&  \hspace{1em}+C\int_{t_{0}/2}^{t}(t-s)^{-\frac{1}{2}}\left(  \left\Vert
f\left(  s\right)  \right\Vert _{L^{2}\left(  m_{1}\right)  }+\left\Vert
\nabla_{x}f\left(  s\right)  \right\Vert _{L^{2}\left(  m_{\theta}\right)
}+\left\Vert g\left(  s\right)  \right\Vert _{L^{2}\left(  m_{1}\right)
}+\left\Vert \nabla_{x}g\left(  s\right)  \right\Vert _{L^{2}\left(
m_{\theta}\right)  }\right)  ds\\
&  \leq C\left(  t-t_{0}/2\right)  ^{-\frac{1}{2}}t_{0}^{-\frac{3}{2}}\left(
\left\Vert f_{in}\right\Vert _{L^{2}\left(  m_{1}\right)  }+\left\Vert
g_{in}\right\Vert _{L^{2}\left(  m_{1}\right)  }\right)  ,
\end{align*}
due to Lemma \ref{hx-regularityST}, $\left(  \ref{fx-energy-BB2}\right)  $ and
$\left(  \ref{L-estimate}\right)  $. It follows that
\[
\left\Vert \nabla_{p}f\left(  t\right)  \right\Vert _{L^{2}\left(  m_{\theta
}\right)  }\lesssim t^{-\frac{1}{2}-\frac{3}{2}}\left(  \left\Vert
f_{in}\right\Vert _{L^{2}\left(  m_{1}\right)  }+\left\Vert g_{in}\right\Vert
_{L^{2}\left(  m_{1}\right)  }\right)  ,\ \ \ \ \ 0<t\leq1.
\]
Since $\partial_{x}^{\alpha}$ commutes with the equation $\left(
\ref{fBB}\right)  $, we also have%
\[
\left\Vert \nabla_{p}\nabla_{x}^{k}f\left(  t\right)  \right\Vert
_{L^{2}\left(  m_{\theta}\right)  }\lesssim t^{-\frac{1}{2}-\frac{3}{2}%
(k+1)}\left(  \left\Vert f_{in}\right\Vert _{L^{2}\left(  m_{k+1}\right)
}+\left\Vert g_{in}\right\Vert _{L^{2}\left(  m_{k+1}\right)  }\right)
,\ \ \ \ \ 0<t\leq1.\text{ }%
\]
For general $k,\ell,$ we apply the induction argument combined with $\left(
\ref{L-estimate}\right)  $ to conclude that%
\[
\left\Vert \nabla_{p}^{\ell}\nabla_{x}^{k}f\left(  t\right)  \right\Vert
_{L^{2}\left(  m_{\theta}\right)  }\lesssim t^{-\frac{\ell}{2}-\frac{3}%
{2}(k+\ell)}\left(  \left\Vert f_{in}\right\Vert _{L^{2}\left(  m_{k+\ell
}\right)  \ }+\left\Vert g_{in}\right\Vert _{L^{2}\left(  m_{k+\ell}\right)
\ }\right)
\]
for $0<t\leq1.$
\end{proof}

\begin{proposition}
\label{large-time-h} Let $f$ be a solution to equation $\left(  \ref{fBB}%
\right)  $ and let $k,\ \ell\in\mathbb{N}\cup\{0\}.\ $Then for $t\geq1,$%
\[
\left\Vert \nabla_{p}^{\ell}\nabla_{x}^{k}f(t)\right\Vert _{L^{2}(m_{\theta}%
)}\leq C\left(  \Vert f_{in}\Vert_{L^{2}\left(  m_{k+\ell}\right)  }%
+t^{1/2}\Vert g_{in}\Vert_{L^{2}\left(  m_{k+\ell}\right)  }\right)  ,
\]
the constant $C$ depending only upon $k$ and $\ell.$
\end{proposition}

\begin{proof}
Define
\[
H^{n}[f]\left(  t\right)  =\sum_{j=0}^{n}\left\Vert \nabla_{p}^{j}%
f(t)\right\Vert _{L^{2}(m_{\theta})}^{2}\ \ \text{and\ \ \ }H_{x}%
^{n}[f]\left(  t\right)  =\sum_{j=0}^{n}\left\Vert \nabla_{p}^{j}\nabla
_{x}f(t)\right\Vert _{L^{2}(m_{\theta})}^{2}.
\]
In view of Lemmas \ref{LBA} and \ref{energy-Guo},
\begin{align*}
\frac{d}{dt}\Vert\nabla_{p}^{\ell}f\Vert_{L^{2}\left(  m_{\theta}\right)
}^{2}  &  \lesssim-\Vert\nabla_{p}^{\ell}f\Vert_{L_{\sigma}^{2}(m_{\theta}%
)}^{2}+C\left(  \sum_{j=0}^{\ell-1}\Vert\nabla_{p}^{j}f\Vert_{L_{\sigma}%
^{2}(m_{\theta})}^{2}+\sum_{j=0}^{\ell-1}\Vert\nabla_{p}^{j}\nabla_{x}%
f\Vert_{L^{2}(m_{\theta})}^{2}\right) \\
&  +C\left(  \Vert f\Vert_{L^{2}}^{2}+\sum_{j=0}^{\ell}\Vert\nabla_{p}%
^{j}g\Vert_{L^{2}(m_{\theta})}^{2}\right)  .
\end{align*}
Hence, following a similar argument in Proposition \ref{energy*}, we deduce
\[
\frac{d}{dt}H^{n}[f]\left(  t\right)  \leq-cH^{n}[f]\left(  t\right)
+C\left(  H_{x}^{n-1}[f]\left(  t\right)  +\Vert f\left(  t\right)
\Vert_{L^{2}}^{2}+\sum_{j=0}^{n}\Vert\nabla_{p}^{j}g\left(  t\right)
\Vert_{L^{2}(m_{\theta})}^{2}\right)
\]
for some constants $c>0$ small and $C>0$ large. It implies that for $t\geq1$
\begin{align*}
H^{n}[f]\left(  t\right)   &  \leq e^{-c(t-1)}H^{n}[f]\left(  1\right)
+Ce^{-ct}\int_{1}^{t}e^{cs}\left(  \Vert f\left(  s\right)  \Vert_{L^{2}}%
^{2}+\sum_{j=0}^{n}\Vert\nabla_{p}^{j}g\left(  s\right)  \Vert_{L^{2}%
(m_{\theta})}^{2}\right)  ds\,\\
&  +Ce^{-ct}\int_{1}^{t}e^{cs}H_{x}^{n-1}[f]\left(  s\right)  ds\,\\
&  \lesssim\left(  \Vert f_{in}\Vert_{L^{2}\left(  m_{n}\right)  }^{2}+t\Vert
g_{in}\Vert_{L^{2}\left(  m_{n}\right)  }^{2}\right)  +e^{-ct}\int_{1}%
^{t}e^{cs}H_{x}^{n-1}[f]\left(  s\right)  ds,
\end{align*}
due to $\left(  \ref{f-energy-BB2}\right)  $ and Proposition
\ref{short-time-h}. Further, applying a similar argument to the same
functional $\mathcal{A}_{N}\left(  t,f\right)  $ defined in Proposition
\ref{energy*} gives
\begin{align*}
\frac{d}{dt}\mathcal{A}_{N}\left(  t,f\right)   &  \leq-c\mathcal{A}%
_{N}\left(  t,f\right) \\
&  +C\left(  \sum_{j=1}^{N+1}\left\Vert \nabla_{x}^{j}f(t)\right\Vert _{L^{2}%
}^{2}+\sum_{j=1}^{N+1}\left\Vert \nabla_{x}^{j}g(t)\right\Vert _{L^{2}%
(m_{\theta})}^{2}+\sum_{j=1}^{N}\sum_{q=1}^{N+1-j}\Vert\nabla_{p}^{j}%
\nabla_{x}^{q}g(t)\Vert_{L^{2}(m_{\theta})}^{2}\right)  ,
\end{align*}
so that%
\begin{align*}
\mathcal{A}_{N}\left(  t,f\right)   &  \leq e^{-c\left(  t-1\right)
}\mathcal{A}_{N}\left(  1,f\right) \\
&  +C\int_{1}^{t}e^{-c\left(  t-s\right)  }\left(  \sum_{j=1}^{N+1}\left\Vert
\nabla_{x}^{j}f(s)\right\Vert _{L^{2}}^{2}+\sum_{j=1}^{N+1}\left\Vert
\nabla_{x}^{j}g(t)\right\Vert _{L^{2}(m_{\theta})}^{2}+\sum_{j=1}^{N}%
\sum_{q=1}^{N+1-j}\Vert\nabla_{p}^{j}\nabla_{x}^{q}g(t)\Vert_{L^{2}(m_{\theta
})}^{2}\right)  ds\\
&  \lesssim\left(  \Vert f_{in}\Vert_{L^{2}\left(  m_{N+1}\right)  }%
^{2}+t\Vert g_{in}\Vert_{L^{2}\left(  m_{N+1}\right)  }^{2}\right)  .
\end{align*}
It implies that for $t\geq1$
\[
H_{x}^{n-1}[f]\left(  t\right)  \lesssim\mathcal{A}_{n-1}\left(  t,f\right)
\lesssim\left(  \Vert f_{in}\Vert_{L^{2}\left(  m_{n}\right)  }^{2}+t\Vert
g_{in}\Vert_{L^{2}\left(  m_{n}\right)  }^{2}\right)  .
\]
As a result, we have
\[
H^{n}[f]\left(  t\right)  \lesssim\left(  \Vert f_{in}\Vert_{L^{2}\left(
m_{n}\right)  }^{2}+t\Vert g_{in}\Vert_{L^{2}\left(  m_{n}\right)  }%
^{2}\right)  ,\
\]
and%
\[
\ \ H_{x}^{n}[f]\left(  t\right)  \lesssim\left(  \Vert f_{in}\Vert
_{L^{2}\left(  m_{n+1}\right)  }^{2}+t\Vert g_{in}\Vert_{L^{2}\left(
m_{n+1}\right)  }^{2}\right)  .
\]
Therefore, for $\ell\in\mathbb{N},$%
\[
\left\Vert \nabla_{p}^{\ell}f(t)\right\Vert _{L^{2}(m_{\theta})}%
\lesssim\left(  \Vert f_{in}\Vert_{L^{2}\left(  m_{\ell}\right)  }%
+t^{1/2}\Vert g_{in}\Vert_{L^{2}\left(  m_{\ell}\right)  }\right)
,\ \ \ \ t\geq1.
\]
On the other hand, for any $k,\ \ell\in\mathbb{N}$ we also have
\[
\left\Vert \nabla_{p}^{\ell}\nabla_{x}^{k}f(t)\right\Vert _{L^{2}(m_{\theta}%
)}\lesssim\sqrt{\mathcal{A}_{k+\ell-1}\left(  t,f\right)  }\lesssim\left(
\Vert f_{in}\Vert_{L^{2}\left(  m_{k+\ell}\right)  }+t^{1/2}\Vert g_{in}%
\Vert_{L^{2}\left(  m_{k+\ell}\right)  }\right)  ,\ \ \ \ t\geq1.
\]
Together with Lemma \ref{hx-regularityST} $\left(  ii\right)  ,$ the proof is completed.
\end{proof}

\subsection{Proof of Theorem \ref{theorem-h}}

Let $k,$ $\ell\in\mathbb{N}\cup\{0\}.$ It follows from Lemma
\ref{hx-regularityST}, Propositions \ref{short-time-h} and \ref{large-time-h} that for $0<t\leq1$%
\[
\left\Vert \nabla_{p}^{\ell}\nabla_{x}^{k}h_{L}\left(  t\right)  \right\Vert
_{L^{2}},\text{ }\left\Vert \nabla_{p}^{\ell}\nabla_{x}^{k}h_{S}\left(
t\right)  \right\Vert _{L^{2}}\lesssim t^{-\frac{\ell}{2}-\frac{3}{2}(k+\ell
)}\left\Vert g_{in}\right\Vert _{L^{2}\left(  w_{k+\ell}\right)  \ },\text{ }%
\]
and for $t\geq1$%
\[
\left\Vert \nabla_{p}^{\ell}\nabla_{x}^{k}h_{L}\left(  t\right)  \right\Vert
_{L^{2}},\text{ }\left\Vert \nabla_{p}^{\ell}\nabla_{x}^{k}h_{S}\left(
t\right)  \right\Vert _{L^{2}}\lesssim\left\Vert g_{in}\right\Vert
_{L^{2}\left(  w_{k+\ell}\right)  \ }.
\]
Owing to the fact that
\[
\left\Vert h_{S}\right\Vert _{L^{2}}\lesssim e^{-Ct}\left\Vert g_{in}%
\right\Vert _{L^{2}\ },
\]
it is easy to derive the regularity estimate on $h_{S}$ via the Sobolev
inequality and the interpolation inequality, i.e., for any $k,\ell
\in\mathbb{N\cup\{}0\}$,%
\begin{align*}
\left\Vert \nabla_{p}^{\ell}\nabla_{x}^{k}h_{S}\right\Vert _{L_{x}^{\infty
}L_{p}^{2}}  &  \leq\left\Vert \nabla_{p}^{\ell}\nabla_{x}^{k}h_{S}\right\Vert
_{L_{p}^{2}L_{x}^{\infty}}\lesssim\sqrt{\left\Vert \nabla_{p}^{\ell}\nabla
_{x}^{k+2}h_{S}\right\Vert _{L^{2}}\left\Vert \nabla_{p}^{\ell}\nabla
_{x}^{k+1}h_{S}\right\Vert _{L^{2}}}\\
&  \lesssim\left\Vert h_{S}\right\Vert _{L^{2}}^{2^{-\left(  k+\ell+2\right)
}}\left(  \left\Vert g_{in}\right\Vert _{L^{2}\left(  w_{k+\ell+2}\right)
\ }\right)  ^{1-2^{-\left(  k+\ell+2\right)  }}\\
&  \lesssim e^{-2^{-\left(  k+\ell+2\right)  }Ct}\left\Vert g_{in}\right\Vert
_{L^{2}\left(  w_{k+\ell+2}\right)  \ },
\end{align*}
for all $t\geq1.$ We complete the estimate for $h_{S}.$

Recall the regularization estimate of the homogeneous part of $h.$ To maintain
the decay rate of $h_{L}$ without imposing higher weights on the initial data,
we further refine the regularization estimate on $h_{\perp0}^{L}$ and
$h_{\perp\perp}^{L}$ before employing the Sobolev inequality and interpolation
trick. Specifically, we prove that for any $k,\ell\in\mathbb{N}\cup\{0\},$
\[
\Vert\nabla_{p}^{\ell}\nabla_{x}^{k}h_{\perp0}^{L}\Vert_{L^{2}}\lesssim
t^{1/2}\Vert g_{in}\Vert_{L^{2}},
\]
and
\[
\left\Vert \nabla_{p}^{\ell}\nabla_{x}^{k}h_{\perp\perp}^{L}(t)\right\Vert
_{L^{2}}\leq Ct^{1/2}\Vert g_{in}\Vert_{L^{2}\left(  w_{\ell}\right)  },
\]
for $t\geq1.$ Now, let $h_{\ast,0}^{L}=h_{00}^{L}+h_{\perp0}^{L}$ and then it
satisfies the equation:%
\begin{equation}\label{eq:improve**}
\left\{
\begin{array}
[c]{l}%
\displaystyle\pa_{t}h_{\ast,0}^{L}+\frac{1}{m_{B}}p\cdot\nabla_{x}h_{\ast
,0}^{L}=L_{BB}h_{\ast,0}^{L}+L_{BA}g_{L,0}\,,\quad(t,x,p)\in({\mathbb{R}}%
^{+},{\mathbb{R}}^{3},{\mathbb{R}}^{3})\,,\\[4mm]%
\displaystyle h_{\ast,0}^{L}(0,x,p)=0.
\end{array}
\right.
\end{equation}
Following the energy estimate in Proposition \ref{large-time-h}, together with
$m_{\theta}=1,$ gives
\begin{align*}
\frac{d}{dt}\mathcal{A}_{N}[h_{\ast,0}^{L}]\left(  t\right)   &
\leq-c\mathcal{A}_{N}[h_{\ast,0}^{L}]\left(  t\right)  \\
&  +C\left(  \sum_{j=1}^{N+1}\left\Vert \nabla_{x}^{j}h_{\ast,0}%
^{L}(t)\right\Vert _{L^{2}}^{2}+\sum_{j=1}^{N+1}\left\Vert \nabla_{x}%
^{j}g_{L,0}(t)\right\Vert _{L^{2}}^{2}+\sum_{j=1}^{N}\sum_{q=1}^{N+1-j}%
\Vert\nabla_{p}^{j}\nabla_{x}^{q}g_{L,0}(t)\Vert_{L^{2}}^{2}\right)  ,
\end{align*}
which implies that
\[
\mathcal{A}_{N}[h_{\ast,0}^{L}]\left(  t\right)  \lesssim t\Vert\left(
g_{L,0}\right)  _{in}\Vert_{L^{2}\left(  w_{N+1}\right)  }^{2}+\int_{1}%
^{t}e^{-c\left(  t-s\right)  }\left(  \sum_{j=1}^{N+1}\left\Vert \partial
_{x}^{j}h_{\ast,0}^{L}(s)\right\Vert _{L^{2}}^{2}+\sum_{j=1}^{N}\sum
_{q=1}^{N+1-j}\Vert\nabla_{p}^{j}\nabla_{x}^{q}g_{L,0}(s)\Vert_{L^{2}}%
^{2}\right)  ds.
\]
Combined%
\[
\Vert\left(  g_{L,0}\right)  _{in}\Vert_{L^{2}\left(  w_{n}\right)  }%
^{2}\lesssim\Vert g_{in}\Vert_{L^{2}}^{2}%
\]
with
\[
\left\Vert \nabla_{x}^{k}h_{\ast,0}^{L}\right\Vert _{L^{2}},\ \ \Vert
\nabla_{p}^{\ell}\nabla_{x}^{k}g_{L,0}\Vert_{L^{2}}\lesssim\Vert g_{in}%
\Vert_{L^{2}},\ \ t\geq1,
\]
it follows that%
\[
\mathcal{A}_{N}[h_{\ast,0}^{L}]\left(  t\right)  \lesssim t\Vert g_{in}%
\Vert_{L^{2}}^{2},\ \ t\geq1,
\]
which implies that
\[
\Vert\nabla_{p}^{\ell}h_{\ast,0}^{L}\Vert_{L^{2}}^{2}\lesssim H^{\ell}\left[
h_{\ast,0}^{L}\right]  \left(  t\right)  \lesssim t\Vert\left(  g_{L,0}%
\right)  _{in}\Vert_{L^{2}\left(  w_{\ell}\right)  }^{2}\lesssim t\Vert
g_{in}\Vert_{L^{2}}^{2},
\]
and%
\[
\Vert\nabla_{p}^{\ell}\nabla_{x}^{k}h_{\ast,0}^{L}\Vert_{L^{2}}\lesssim
\sqrt{\mathcal{A}_{\ell+k-1}[h_{\ast,0}^{L}]\left(  t\right)  }\lesssim
t^{1/2}\Vert g_{in}\Vert_{L^{2}}\,,
\]
for all $t\geq1.$ Hence,%
\[
\Vert\nabla_{p}^{\ell}\nabla_{x}^{k}h_{\perp0}^{L}\Vert_{L^{2}}\leq\Vert
\nabla_{p}^{\ell}\nabla_{x}^{k}h_{\ast,0}^{L}\Vert_{L^{2}}+\Vert\nabla
_{p}^{\ell}\nabla_{x}^{k}h_{00}^{L}\Vert_{L^{2}}\lesssim t^{1/2}\Vert
g_{in}\Vert_{L^{2}}\,.
\]
By the Sobolev inequality, we obtain
\[
\Vert\nabla_{p}^{\ell}\nabla_{x}^{k}h_{\perp0}^{L}\Vert_{L_{x}^{\infty}%
L_{p}^{2}}\lesssim\Vert\nabla_{p}^{\ell}\nabla_{x}^{k+2}h_{\perp0}^{L}%
\Vert_{L^{2}}^{1/2}\Vert\nabla_{p}^{\ell}\nabla_{x}^{k+1}h_{\perp0}^{L}%
\Vert_{L^{2}}^{1/2}\lesssim t^{1/2}\Vert g_{in}\Vert_{L^{2}}.
\]
Consequently, employing the interpolation inequality in $p$
(\textbf{NOT} in $x$) yields%
\begin{align}
\left\Vert \nabla_{p}^{\ell}\nabla_{x}^{k}h_{\perp0}^{L}\right\Vert
_{L_{x}^{\infty}L_{p}^{2}} &  \lesssim\left\Vert \nabla_{x}^{k}h_{\perp0}%
^{L}\right\Vert _{L_{x}^{\infty}L_{p}^{2}}^{1/2}\left\Vert \nabla_{p}^{2\ell
}\nabla_{x}^{k}h_{\perp0}^{L}\right\Vert _{L_{x}^{\infty}L_{p}^{2}}%
^{1/2}\nonumber\\
&  \lesssim\left\Vert \nabla_{x}^{k}h_{\perp0}^{L}\right\Vert _{L_{x}^{\infty
}L_{p}^{2}}^{\left(  1/2+1/4\right)  }\left\Vert \nabla_{p}^{4\ell}\nabla
_{x}^{k}h_{\perp0}^{L}\right\Vert _{L_{x}^{\infty}L_{p}^{2}}^{1/4}\nonumber\\
&  \lesssim\left\Vert \nabla_{x}^{k}h_{\perp0}^{L}\right\Vert _{L_{x}^{\infty
}L_{p}^{2}}^{\left(  1/2+1/4+\cdots+1/2^{q}\right)  }\left\Vert \nabla
_{p}^{2^{q}\ell}\nabla_{x}^{k}h_{\perp0}^{L}\right\Vert _{L_{x}^{\infty}%
L_{p}^{2}}^{2^{-q}}\nonumber\\
&  \lesssim\left(  1+t\right)  ^{-\frac{4+k}{2}\left(  1-2^{-q}\right)  }%
\cdot t^{2^{-q-1}}\left(  \left\Vert g_{in}\right\Vert _{L^{2}\ }+\Vert g_{in}\Vert
_{L_{x}^{1}L_{p}^{2}}\right)  \nonumber\\
&  =\left(  1+t\right)  ^{-\frac{k+3}{2}+2^{-q-1}\left(  5+k-2^{q}\right)
}\cdot\left(  \left\Vert g_{in}\right\Vert _{L^{2}\ }+\Vert g_{in}\Vert
_{L_{x}^{1}L_{p}^{2}}\right)  \nonumber\\
&  \leq\left(  1+t\right)  ^{-\frac{k+3}{2}}\cdot\left(  \left\Vert
g_{in}\right\Vert _{L^{2}\ }+\Vert g_{in}\Vert_{L_{x}^{1}L_{p}^{2}}\right)
,\label{I3-regularity}%
\end{align}
whenever $q\geq\left\lceil \log_{2}\left(  5+k\right)  \right\rceil .$

Finally, let $h_{\ast,\perp}^{L}=h_{0\perp}^{L}+h_{\perp\perp}^{L}$ and then
it satisfies the equation:
\begin{equation}
\left\{
\begin{array}
[c]{l}%
\displaystyle\pa_{t}h_{\ast,\perp}^{L}+\frac{1}{m_{B}}p\cdot\nabla_{x}%
h_{\ast,\perp}=L_{BB}h_{\ast,\perp}^{L}+L_{BA}g_{L,\perp}\,,\quad
(t,x,p)\in({\mathbb{R}}^{+},{\mathbb{R}}^{3},{\mathbb{R}}^{3})\,,\\[4mm]%
\displaystyle h_{\ast,\perp}^{L}(0,x,p)=0.
\end{array}
\right.
\end{equation}
Going through the proofs of Propositions \ref{short-time-h} and
\ref{large-time-h} with the fact that $\left\Vert \nabla_{p}^{\ell}\nabla
_{x}^{k}g_{L,\perp}(t)\right\Vert _{L^{2}}$ $\leq C\Vert g_{in}\Vert
_{L^{2}\left(  w_{\ell}\right)  }$, we find
\[
\left\Vert \nabla_{p}^{\ell}\nabla_{x}^{k}h_{\ast,\perp}^{L}(t)\right\Vert
_{L^{2}}\leq Ct^{1/2}\Vert g_{in}\Vert_{L^{2}\left(  w_{\ell}\right)
},\ \ t\geq1.
\]
Since
\[
\left\Vert \nabla_{p}^{\ell}\nabla_{x}^{k}h_{0\perp}^{L}(t)\right\Vert
_{L^{2}}\leq C\Vert g_{in}\Vert_{L^{2}},
\]
we improve
\[
\left\Vert \nabla_{p}^{\ell}\nabla_{x}^{k}h_{\perp\perp}^{L}(t)\right\Vert
_{L^{2}}\leq Ct^{1/2}\Vert g_{in}\Vert_{L^{2}\left(  w_{\ell}\right)
},\ \ t\geq1.
\]
Therefore,
\begin{align*}
\left\Vert \nabla_{p}^{\ell}\nabla_{x}^{k}h_{\perp\perp}^{L}\right\Vert
_{L_{x}^{\infty}L_{p}^{2}}  &  \lesssim\sqrt{\left\Vert \nabla_{p}^{\ell
}\nabla_{x}^{k+2}h_{\perp\perp}^{L}\right\Vert _{L^{2}}\left\Vert \nabla
_{p}^{\ell}\nabla_{x}^{k+1}h_{\perp\perp}^{L}\right\Vert _{L^{2}}}\\
&  \lesssim\sqrt{\left\Vert \nabla_{x}^{k+1}h_{\perp\perp}^{L}\right\Vert
_{L^{2}}^{2^{-\ell}}\left(
{\displaystyle\prod\limits_{j=2}^{\ell}}
\left\Vert \nabla_{p}^{j}\nabla_{x}^{k+1}h_{\perp\perp}^{L}\right\Vert
_{L^{2}}^{2^{-\left(  \ell-j+2\right)  }}\right)  \left\Vert \nabla_{p}%
^{\ell+1}\nabla_{x}^{k+1}h_{\perp\perp}^{L}\right\Vert _{L^{2}}^{1/2}%
\left\Vert \nabla_{p}^{\ell}\nabla_{x}^{k+2}h_{\perp\perp}^{L}\right\Vert
_{L^{2}}}\\
&  \lesssim e^{-2^{-\left(  \ell+2\right)  }a(\tau)t}\Vert g_{in}\Vert
_{L^{2}\left(  w_{\ell+1}\right)  },
\end{align*}
This completes the proof of Theorem \ref{theorem-h}.

\end{document}